\pdfoutput=1
\RequirePackage{ifpdf}
\ifpdf 
\documentclass[pdftex]{sigma}
\else
\documentclass{sigma}
\fi

\numberwithin{equation}{section}

\providecommand{\norm}[2][\relax]{\left\|#2\right\|\ifx#1\relax\else_{#1}\fi}
\providecommand{\modulus}[2][\relax]{\left| #2 \right|\ifx#1\relax\else_{#1}\fi}
\providecommand{\oper}[1]{\mathcal{#1}}
\providecommand{\algebra}[1]{\ensuremath{\mathfrak{#1}}}
\providecommand{\Space}[3][]{\ifx#2R\ifx#1e \mathbb{C}^{#3} \else
\ifx#1p \mathbb{D}^{#3} \else
\ifx#1h \mathbb{O}^{#3} \else
\ifx#1\sigma \mathbb{A}\!^{#3} \else
\ensuremath{\mathbb{#2}^{#3}_{#1}{}} \fi \fi \fi \fi \else
\ensuremath{\mathbb{#2}^{#3}_{#1}{}} \fi}
\providecommand{\FSpace}[3][]{\ensuremath{\ifx#2l \ell_{#3}^{#1}{}\else
 \mathsf{#2}_{#3}^{#1}{}\fi}}
\providecommand{\scalar}[3][\relax]{\left\langle #2,#3
 \right\rangle\ifx#1\relax\else_{#1}\fi}
\providecommand{\SL}[1][2]{\ensuremath{\mathrm{SL}_{#1}(\Space{R}{})}}
\providecommand{\rmi}{\mathrm{i}}
\providecommand{\rme}{\mathrm{e}}
\providecommand{\rmd}{\mathrm{d}}

\providecommand{\linv}[2][\relax]{\mathfrak{L}^{#2}_{#1}}
\providecommand{\myhbar}{\hslash}

\providecommand{\map}[1]{\mathsf{#1}}
\providecommand{\uir}[3][0]{\ifcase #1{\rho^{#2}_{#3}}%
\or {\breve{\rho}^{#2}_{#3}}%
\or {\tilde{\rho}^{#2}_{#3}}\fi}

\newtheorem{thm}{Theorem}[section]
 \newtheorem{prop}[thm]{Proposition}
 \newtheorem{lem}[thm]{Lemma}
 \newtheorem{cor}[thm]{Corollary}
 \theoremstyle{definition}
 \newtheorem{defn}[thm]{Definition}
 \newtheorem{rem}[thm]{Remark}
 \newtheorem{Example}[thm]{Example}

\begin{document}
\allowdisplaybreaks

\newcommand{\arXivNumber}{2105.13811}

\renewcommand{\PaperNumber}{065}

\FirstPageHeading

\ShortArticleName{Tuning Co- and Contra-Variant Transforms: the Heisenberg Group Illustration}

\ArticleName{Tuning Co- and Contra-Variant Transforms:\\ the Heisenberg Group Illustration}
\Author{Amerah A. AL AMEER~$^{\rm a}$ and Vladimir V.~KISIL~$^{\rm b}$}

\AuthorNameForHeading{A.A.~Al Ameer and V.V.~Kisil}

\Address{$^{\rm a)}$~School of Science, Mathematics Department, University of Hafr Al Batin,\\
\hphantom{$^{\rm a)}$}~Hafr Al Batin 31991\, P.O Box 1803, Saudi Arabia}
\EmailD{\href{mailto:aamalameer@uhb.edu.sa}{aamalameer@uhb.edu.sa}}

\Address{$^{\rm b)}$~School of Mathematics, University of Leeds, Leeds LS2\,9JT, UK}
\EmailD{\href{mailto:V.Kisil@leeds.ac.uk}{V.Kisil@leeds.ac.uk}}
\URLaddressD{\url{https://www1.maths.leeds.ac.uk/~kisilv/}}

\ArticleDates{Received December 26, 2021, in final form August 26, 2022; Published online September 01, 2022}

\Abstract{We discuss a fine tuning of the \emph{co- and contra-variant transforms} through construction of specific fiducial and reconstructing vectors. The technique is illustrated on three different forms of induced representations of the Heisenberg group. The covariant transform provides intertwining operators between pairs of representations. In particular, we obtain the Zak transform as an induced covariant transform intertwining the Schr\"odinger representation on $\FSpace{L}{2}(\Space{R}{})$ and the lattice (nilmanifold) representation on $\FSpace{L}{2}\big(\Space{T}{2}\big)$. Induced covariant transforms in other pairs are Fock--Segal--Bargmann and theta transforms. Furthermore, we describe \emph{peelings} which map the group-theoretical induced representations to convenient representation spaces of analytic functions. Finally, we provide a condition which can be imposed on the reconstructing vector in order to obtain an intertwining operator from the induced contravariant transform.}


\Classification{43A85; 47G10; 81R30}

\section{Introduction}

The purpose of this paper is to present an advanced use of the \emph{induced co- and contra-variant transform}, which is created by the Gilmore--Perelomov coherent states~(see \cite{Perelomov86} and \cite[Section~7.1]{AliAntGaz14a}). The transform is an intertwining operator to an induced representation, which explains our choice of the name for it. The approach is illustrated here by the crucial example of the Heisenberg group $\Space{H}{1}$, however the technique is not limited to this case, cf.~\cite{AlmalkiKisil19a,Kisil11c,Kisil12d}. The topics of coherent states and covariant transform (also known under many other names) are extensively covered in the existing literature, e.g., \cite[Section~13]{Kirillov76}, \cite[Appendix~V.2]{Kirillov04a}, and \cite{Berezin86,Grossmann85a,Grossmann86a,HernandezSikicWeissWilson11a, Kisil12b, Kisil17a,Perelomov86}, and we refer to authoritative surveys \cite{AliAntGaz14a,Folland89,HernandezLuthySikicSoriaWilson21a} for further references. Our purpose is to present some additional aspects which are commonly shadowed or missing in the existing sources. If these properties are explicitly stated then many known important results immediately follow as their direct corollaries.

For example, the standard induction (see \cite[Chapter~6]{Folland16a},
\cite[Appendix~V.2]{Kirillov04a} and \cite{BarbieriHernandezMayeli14a,HernandezLuthySikicSoriaWilson21a, KaniuthTaylor13a})  from a~character of the Heisenberg group centre gives the representation~\eqref{eq:H1-pol-rep-plane} below, which is different from the commonly used celebrated Fock--Segal--Bargmann (FSB) representation~\eqref{FSB-Rep} in the space of analytic function~\cite[Section~1.6]{Folland89}. Those two representations are linked by the \emph{peeling} map which transforms the annihilator of the representation space to the Cauchy--Riemann operator, see Section~\ref{sec: Peeling}. The origin of the annihilator operator is revealed as the Lie derivative with a special relation to the chosen fiducial vector, see Section~\ref{Covariant Transform}. Section~\ref{Contravariant Transform} provides the respective consideration of the reconstructing vector for the contravariant transform.

In this paper we present a machinery which allows to design fiducial and reconstructing vectors to ensure specific properties within the induced co- and contra-variant transforms. This technique does not rely on \emph{ad hoc} knowledge and allows one to get new insights even within the much-studied framework of the Heisenberg group \cite{BarbieriHernandezMayeli14a,HernandezSikisWeissWilson10a}. As an basic illustration, we apply this technique to find contra- and co-variant transforms, which intertwine the coordinate and momentum representations of $\Space{H}{1}$, and predictably obtain the Fourier transform and its inverse, see Example~\ref{ex:covariant-Fourier}. Less elementary example is the interpretation of the Zak transform as an induced covariant transform in Theorem~\ref{th:Zak}, which emerges as follows.

There are three forms of induction of representations of the Heisenberg group~\cite[Section~2.2]{Kirillov04a}: the \emph{left quasi-regular} representation, the \emph{Schr\"odinger} representation and the \emph{lattice} representation (see Section~\ref{sec:induc-repr-heis}) for details. We systematically and uniformly use the covariant transform for them. In particular,
 the Zak transform and its inverse are expressed as the covariant transform between the Schr\"odinger and lattice representations, with the \emph{Jacobi theta function} appearing as a vacuum state of the latter, see Theorem~\ref{th:Zak}.
 Similarly, expressing the \emph{pre-theta transform and its inverse} throughout the same technique is also new, see Theorems~\ref{theo_Thata} and~\ref{th:inverse-Zak-contravar}.
The pre-theta transform and its inverse intertwine the (pre-)FSB and the lattice representations. The \emph{$($pre-$)$ Fock--Segal--Bargamann $($FSB$)$ transform} and its inverse (see \cite[Section~4.2]{Neretin11a} and~\cite{Folland89}) intertwines the Schr\"odinger representation on $\FSpace{L}{2}(\Space{R}{})$ and left quasi-regular representation on $\FSpace{L}{2}\big(\Space{R}{2}\big)$. We name it as the FSB transform from quantum mechanics (see \cite[Section~4.2]{Neretin11a} and~\cite{Folland89}), it is also known as the Gabor or time-frequency transform or windowed Fourier transform of a~signal~\cite{Grochenig01a, Massopust14a}.

The classical Zak transform~\cite{Zak67a}, also known as Weil--Brezin transform~\cite[Section~1.10]{Folland89}, can be traced back to the works of Gelfand in 1950, see also \cite[Section~1.5]{Perelomov86}, \cite[Chapter~8]{Grochenig01a}, \cite[Section~8.1]{Neretin11a} and \cite{HernandezLuthySikicSoriaWilson21a,HernandezSikisWeissWilson10a,HernandezSikicWeissWilson11a} for further applications and historical notes.
 This transform is an isometric isomorphism from $\FSpace{L}{2}(\Space{R}{})$ onto $\FSpace{L}{2}\big(\Space{T}{2}\big)$ given, for $f\in\FSpace{L}{2}(\Space{R}{})$, by
 \begin{equation*}
	\big[\oper{\tilde Z}f\big](u,v)=\sum_{n\in\Space{Z}{}} f(u+n)\,\rme^{2\pi\rmi m n v}, \qquad (u,v)\in\Space{T}{2}.
	\end{equation*}
Weil~\cite{Weil64a} defined the abstract Zak transform on arbitrary locally compact abelian (LCA) groups with respect to arbitrary closed subgroups.
Subsequently, the Zak transform was reviewed and generalised by many authors, see for example~\cite{ArefijamaalGhaani13a,BarbieriHernandezPaternostro15a,HernandezLuthySikicSoriaWilson21a,HernandezSikicWeissWilson11a, Iverson19a}. Various connections between the Zak transform and (left) shift-invariant spaces were studied by several authors, cf.~\cite{BarbieriHernandezMayeli14a,HernandezSikisWeissWilson10a, Iverson15a,Iverson18a}, further references may be found in the recent survey paper~\cite{HernandezLuthySikicSoriaWilson21a}. Yet, an explicit interpretation of the Zak transform and its inverse as a co- and contravariant transforms respectively, cf.\ Theorems~\ref{th:Zak} and~\ref{th:inverse-Zak-contravar}, appears to be new in this paper.

\section[Preliminaries on the Heisenberg group and its induced representations]{Preliminaries on the Heisenberg group\\ and its induced representations}\label{sec:repr-homog-spac}

\subsection{The Heisenberg group and its Lie algebra}\label{sec:heosenberg-group-its}
The \emph{polarised Heisenberg group} $\Space[p]{H}{n}$ is the set of triples $(s,x,y)$, where $s\in\Space{R}{}$ and $x$, $y\in \Space{R}{n}$, with the group law given as follows~\cite[Section~1.2]{Folland89}
\begin{equation*} 
 (s,x,y)\cdot(s',x',y')=(s+s'+xy',x+x',y+y').
\end{equation*}
For the sake of simplicity in this paper, we will work with the one-dimensional case
 of $\Space[p]{H}{1}$ and call it the Heisenberg group. It is a non-commutative group and its centre is a one-dimensional subgroup
\begin{equation} \label{eq:heisenberg-centre}
 Z=\big\{(s,0,0)\in \Space[p]{H}{1}\colon s \in \Space{R}{}\big\} .
\end{equation}

The left action of $\Space[p]{H}{1}$ on itself is given by
 \begin{equation*}
\tilde\Lambda(g)\colon \ g'\mapsto g^{-1} g'.
\end{equation*}
 We extend this action to a linear representation
\begin{equation}
 \label{eq:left-right-regular}
 \Lambda(g)\colon \ F(g') \mapsto F\big(g^{-1}g'\big), \qquad g,g'\in\Space[p]{H}{1}
\end{equation}
on a certain linear space of functions on $\Space[p]{H}{1}$~\cite[Section~1.1]{Berndt07a}.
The Lebesgue measure ${\rm d}g={\rm d}s\,{\rm d}x\,{\rm d}y$ on
 $\Space[p]{H}{1}\sim\Space{R}{3}$ is a \emph{Haar measure} invariant under the left and right shifts. The action~\eqref{eq:left-right-regular} on the Hilbert space $\FSpace{L}{2}\big(\Space[p]{H}{1},{\rm d}g\big)$ of square-integrable functions on $\Space[p]{H}{1}$
is unitary~\cite[Appendix~V.2]{Kirillov04a} and is called the \emph{left regular} representation.

Here, we introduce
\begin{equation*}
	S=(1,0,0),\qquad X=(0,1,0),\qquad Y=(0,0,1),
\end{equation*}
which forms the basis of the Lie algebra $\algebra{h}_1$ of $\Space[p]{H}{1}$.
The commutator of ${X}$ and ${Y}$ is given by the celebrated \emph{Heisenberg commutation relation}
\begin{equation} \label{eq:heisenberg-commutator}
	[{X},{Y}]= {S}.
\end{equation}

It is common for a representation $\uir{}{}$ of a Lie group to pass to the derived representations $\rmd \uir{}{}$ of the respective Lie algebra (cf.~\cite[Chapter~2]{Kirillov04a}).
Consider the derived representations of $\algebra{h}_1$ spanned by
 \begin{equation*}
 \rmd\uir{S}{}=\frac{\rmd }{\rmd t} \uir{}{}\big(\rme^{t S}\big)\bigg|_{t=0},\qquad \rmd\uir{X}{}=\frac{\rmd }{\rmd t}\uir{}{}\big(\rme^{t X}\big)\bigg|_{t=0},\qquad \rmd\uir{Y}{}=\frac{\rmd }{\rmd t}\uir{}{}\big(\rme^{t Y}\big)\bigg|_{t=0}.
\end{equation*}

If $\uir{}{}$ is irreducible, $\rmd\uir{S}{}$ is a multiple of the identity operator $I$; that is, $\rmd\uir{S}{}=-\rmi\myhbar I$ (cf.~\cite[Section~2.2]{Kirillov04a}, \cite[Section~II.3]{Bohm93}).

The important operators produced by the derived representations of $\algebra{h}_1$ are \emph{ladder operators} (cf.\ \cite[Section~5.3]{Alamer19a}, \cite[Section~II.3]{Bohm93}, \cite[Section~2.2.1]{Kirillov04a}).
\begin{defn}
Let $\kappa>0$ be some fixed number and $\uir{}{}$ be a representation of $\Space[p]{H}{}$ such that $\uir{}{}(s,0,0)=\rme^{2\pi\rmi \myhbar s}I$. The \emph{ladder operators} are defined as follows
	\begin{equation*}
 {a^-}{}=\frac{1}{\sqrt{2\myhbar\kappa}}\big(\kappa \rmi\,\rmd\uir{X}{}- \rmd\uir{Y}{}\big),\qquad
 {a^+}{}=\frac{1}{\sqrt{2\myhbar\kappa}}\big(\kappa\rmi\,\rmd\uir{X}{} + \rmd\uir{Y}{}\big),
\qquad \myhbar>0 .
	\end{equation*}
 The operators $a^+$ and $a^-$ are known as the {\em creation} and {\em annihilation} operators, respectively.
 \end{defn}
In this paper, we fix the parameter $\kappa>0$ of the ladder operators $a^{\pm}$~\cite{AlmalkiKisil18a} and indicate the dependence of ladder operators upon it. The Heisenberg commutator relation~\eqref{eq:heisenberg-commutator} implies
\begin{equation*}
 [a^-, a^+]=a^-a^+-a^+a^-=I.
\end{equation*}
Also, for a unitary representation $\uir{}{}$, we have $\big(a^{-1}\big)^*=a^+$.

\begin{defn}[{\cite[Section~5.3]{Alamer19a} and \cite[Section~2.6]{Kirillov04a}}] \label{de: vacuum}
 In the above notations, a vector $\phi_0\in\FSpace{H}{}$ is called a {\em vacuum vector} if it is a null solution of the annihilation operator:
 \begin{displaymath}
 a^-\phi_0=0.
 \end{displaymath}
\end{defn}
We will need the following properties of the vacuum (see~\cite[Section~2.6]{Kirillov04a}):
\begin{lem} \label{lem:vacuum-basic}
 If $\phi_0$ is a vacuum of an irreducible representation $\uir{}{}$, then
 \begin{enumerate}\itemsep=0pt
 \item[$1.$] 
 $\phi_0$ is unique up to a scalar multiple.
 \item[$2.$] 
 For an intertwining operator $\oper{W}$ between $\uir{}{}$ and another representation $\uir{}{1}$, the image $\oper{W}\phi_0$ is a vacuum for $\uir{}{1}$.
 \end{enumerate}
\end{lem}

 \subsection{Induced representations of the Heisenberg group}\label{sec:induc-repr-heis}

 We are interested in three representations of $\Space[p]{H}{1}$, which are \emph{induced by characters} of certain subgroups of $\Space[p]{H}{1}$. Here, we briefly set the notations for this particular case of a general theory of representations induced in the sense of Mackey (the reader is referred to~\cite{Folland16a,KaniuthTaylor13a, Kirillov76} for detailed presentations).

Let $H$ be a subgroup of $\Space[p]{H}{1}$ and $\chi$ be a (complex unitary) character of $H$. Let $X=\Space[p]{H}{1}/H$ be the corresponding left $\Space[p]{H}{1}$-homogeneous space with the measure $d x$, which factorises the Haar measure ${\rm d}g={\rm d}x\,{\rm d}h$ for a Haar measure ${\rm d}h$ of $H$. We write $\FSpace[\chi]{L}{2}\big(\Space[p]{H}{1}\big)$ for the space of functions~$F(g)$ on~$\Space[p]{H}{1}$ having these properties:
\begin{enumerate}
\item[(C)] \emph{$H$-covariance}
 \begin{equation*}
 F(gh)=\bar{\chi}_\myhbar (h)F(g), \qquad \text{for all} \quad g\in \Space[p]{H}{1}, \quad h\in H
\end{equation*}
\item[(S)] \emph{$\FSpace{L}{2}$-summability} over $X$
\begin{equation*}
 \int_{X} \modulus{F(g)}^2\, {\rm d}x<\infty,
\end{equation*}
where the integral is meaningful since $\modulus{F(g)}=\modulus{F(gh)}$ for all $h\in H$ by $H$-covariance of~$F$.
\end{enumerate}
The space $\FSpace[\chi]{L}{2}\big(\Space[p]{H}{1}\big)$ is invariant under the left $\Space[p]{H}{1}$-shifts~\eqref{eq:left-right-regular} because the left and right shifts commute. An alternative realisation of the same representation is obtained from a given section $\map{s}\colon \Space[p]{H}{1}/H \rightarrow \Space[p]{H}{1}$ which is a right inverse for the quotient map $\map{p}\colon \Space[p]{H}{1} \rightarrow \Space[p]{H}{1}/H $; that is, \mbox{$\map{p}(\map{s}(w))=w$} for all $w\in \Space[p]{H}{1}/H $. For such a section $\map{s}$ and the character~$\chi$, we can define the lifting $\oper{L}_\chi\colon \FSpace{L}{2}\big(\Space[p]{H}{1}/H\big) \rightarrow \FSpace[\chi]{L}{2}\big(\Space[p]{H}{1}\big)$:
\begin{equation*} 
 [\oper{L}_\chi f](g) = \bar{\chi}(r) f(\map{p}(g)),\qquad \text{where} \quad r = (\map{s}(\map{p}(g)))^{-1} g.
\end{equation*}
Then, the lifting $\oper{L}_\chi $ intertwines a representation $\uir{}{\chi}$ on $\FSpace{L}{2}\big(\Space[p]{H}{1}/H\big) $ and the left regular representation $\Lambda$~\eqref{eq:left-right-regular} restricted to $\FSpace[\chi]{L}{2}\big(\Space[p]{H}{1}\big)$:
\[
\oper{L}_\chi \circ \uir{}{\chi}(g) = \Lambda(g)|_{\FSpace[\chi]{L}{2}(\Space[p]{H}{1})}\circ \oper{L}_\chi, \qquad \text{for all} \quad g \in \Space[p]{H}{1}.
\]
Combining the previous identities, we obtain an explicit expression for the induced representation
\begin{equation} \label{eq:def-ind}
 [\uir{}{\chi}(g) f](x)= \bar{\chi}\big(\map{r}\big(g^{-1} * \map{s}(x)\big)\big)
 f\big(g^{-1}\cdot x\big), \qquad \text{where} \quad f(x) \in \FSpace{L}{2}\big(\Space[p]{H}{1}/H\big).
\end{equation}

We now provide three forms of induced representations of $\Space[p]{H}{1}$ acting on Hilbert spaces of square-integrable functions on homogeneous spaces, where the inductions are performed from characters of subgroups~$H$ of $\Space[p]{H}{1}$ (cf.~\cite[Section~5.1]{Alamer19a}).
\begin{enumerate}\itemsep=0pt
\item For the centre $Z$~\eqref{eq:heisenberg-centre} of $\Space[p]{H}{1}$,
 the homogeneous space $\Space[p]{H}{1}/Z$ is isomorphic to $ \Space{R}{2}$.
 Any real number $\myhbar\neq 0$ defines
 the non-trivial character $\chi_\myhbar(s,0,0)={\rm e}^{2\pi \rmi \myhbar s}$ of $Z$.
 The respective maps are
 \begin{equation*} 
 \map{p}(s,x,y)= (x,y), \qquad
 \map{s}(x,y)= (0,x,y), \qquad
 \map{r}(s,x,y)= (s,0,0).
 \end{equation*}

 The representation $\Lambda_{\myhbar}$, which is induced from $\chi_\myhbar$ on the space of square-integrable functions $\FSpace{L}{2}\big(\Space{R}{2}\big)$, is given by
\begin{equation}\label{eq:H1-pol-rep-plane}
 	 [\Lambda_{\myhbar}(s,x,y)f] (x',y') = {\rm e}^{2\pi \rmi \myhbar (s+x(y'-y))}f(x'-x,y'-y).	
 \end{equation}
 This is called the \emph{left quasi-regular} representation. It is unitary reducible, and it can be decomposed into unitary irreducible components in many different ways (see~\cite[Section~6.4]{Alamer19a}). For physical (aka mathematical) reasons, the most popular irreducible module is the \emph{pre-Fock--Segal--Bargmann space} (\emph{pre-FSB space}):
\begin{equation}\label{pre-FSB}
	\FSpace{F}{\phi_{\myhbar\kappa}}\big(\Space{R}{2}\big)=\big\{{f}\colon \big(\kappa \rmi \linv{X}{}+ \linv{Y}{}\big) f = 0 \text{ and } f\in \FSpace{L}{2}(\Space{R}{}) \big\},
 \end{equation}
 where the \emph{Lie derivative} $\linv{A}=\rmd R_{\myhbar}^{A}$ is the derivation of the \emph{right quasi-regular} representation~\cite[Section~5.4.2]{Alamer19a} on $\FSpace{L}{2}\big(\Space{R}{2}\big)$:
 \begin{equation}\label{right-action-center-rep}
 [R_\myhbar(g)f](x',y')=\rme^{-2\pi\rmi\myhbar(s+x'y)}f(x'+x,y'+y).
 \end{equation}
 We will explain the origin of pre-FSB space in Section~\ref{Pre-Fock-Segal-Bargmann Transform}.

	\item For a two-dimensional maximal Abelian continuous subgroup of
 $\Space[p]{H}{1}$:
	\begin{equation*}
 H'_x=\big\{(s,0,y)\in\Space[p]{H}{1}\colon s,y\in\Space{R}{}\big\},
 \end{equation*}
 the homogeneous space $\Space[p]{H}{1}/ H'_x$ can be identified with $\Space{R}{}$.
 The respective maps are
 \begin{equation} \label{eq:maps-p-s-r-H_x-prime}
 \map{p}(s,x,y)=x, \qquad
 \map{s}(t) = (0,t,0),\qquad
 \map{r}(s,x,y) = (s-xy,0,y).
 \end{equation}
 The character $\chi_\myhbar(s,0,y)=\rme^{2\pi \rmi \myhbar s}$ of $ H'_x$ induces the representation $\uir{}{\myhbar}$ on a space of square-integrable functions, $\FSpace{L}{2}(\Space{R}{})$, is given by
 \begin{equation} \label{eq:H1-schroedinger-rep}
[\uir{}{\myhbar}(s,x,y)f](t)= \rme^{2\pi \rmi \myhbar
 (s-ty)} f(t-x).
 \end{equation}
 This is the celebrated \emph{Schr\"odinger} representation, which is a unitary irreducible representation on $\FSpace{L}{2}(\Space{R}{})$~\cite[Chapter~1]{Folland89}.

 There is another two-dimensional maximal Abelian continuous subgroup of
 $\Space[p]{H}{1}$:
 \begin{equation*} 
 H'_y=\big\{(s,x,0)\in\Space[p]{H}{1}\colon s,x\in\Space{R}{}\big\}.
 \end{equation*}
 Subgroups $H'_x$ and $H'_y$ are conjugated by the automorphism of $\map{i}\colon \Space[p]{H}{1} \rightarrow \Space[p]{H}{1}$:
 \begin{equation*} 
 \map{i}(s,x,y)=(s-xy,-y,x).
 \end{equation*}
 There exist small but important differences between the respective maps cf.~\eqref{eq:maps-p-s-r-H_x-prime}:
 \begin{equation} \label{eq:maps-p-s-r-H_y-prime}
 \map{p}(s,x,y)=y, \qquad
 \map{s}(\lambda ) = (0,0,\lambda ),\qquad
 \map{r}(s,x,y) = (s,x,0).
 \end{equation}
 The character $\chi_{\myhbar}(s,x,0)=\rme^{2\pi \rmi\myhbar s}$ of $H_y'$ induced an alternative form of the Schr\"odinger representation
\begin{equation}
 \label{eq:H1-schroedinger-rep-alt}
 \textstyle[\uir{\prime}{\myhbar}(s,x,y)f](\lambda )= \rme^{2\pi \rmi \myhbar
 (s+x(\lambda -y))}\, f(\lambda -y).
 \end{equation}
 In quantum mechanics, \eqref{eq:H1-schroedinger-rep} and~\eqref{eq:H1-schroedinger-rep-alt} serve as \emph{coordinate} and \emph{momentum} representations of~$\Space[p]{H}{1}$.
\item For a non-commutative discontinuous subgroup $H_d$ of $\Space[p]{H}{1}$:
	\begin{equation}\label{non-commuta}
H_d=\big\{(s, n, k):=(s ,n+\rmi k)\colon (n, k)\in \Space{Z}{2}, s\in\Space{R}{}\big\}
\end{equation}
 we consider a homogeneous space $\Space{T}{2}=\Space[p]{H}{1}/ H_d$.
 Let $\chi_m(s,0,0)={\rm e}^{2\pi \rmi m s}$ be a character of~$ H_d$.
 Let $f(u,v)$ be a function of $\FSpace{L}{2}\big(\Space{T}{2}\big)$. We have two possibilities to treat~$f$:
 \begin{enumerate}\itemsep=0pt
 \item $f$ is a square-integrable and double quasi-periodic function on $\Space{R}{2}$ (periodic in $u$ and quasi-periodic in~$v$), that is, for all $(n, k)\in \Space{Z}{2}$, we have
 \begin{equation} \label{eq: H_d-covariance property}
 f(u+n ,v+k)=\rme^{2\pi m\rmi u k}f(u,v);	
 \end{equation}
\item $f$ is a square-integrable function on the torus $\Space{T}{2}=\lbrace (u,v)\colon u, v\in[0,1) \rbrace$.
\end{enumerate}
If $f$ is considered as a double quasi-periodic function on $\Space{R}{2}$, the representation $\uir{}{m}$ induced from the character $\chi_m$ of $ H_d$ on $\FSpace{ L}{2}\big(\Space{T}{2}\big)$ is given by
\begin{equation} \label{induced-smooth-rep-from-H_d}
 [\uir{}{m}(s,x,y)f](u,v)=\rme^{2\pi m\rmi(s+x(v-y))}f(u-x,v-y).
\end{equation}
The representation $\uir{}{m}$ is unitary irreducible on $\FSpace{ L}{2}\big(\Space{T}{2}\big)$ and is called the \emph{lattice} representation~\cite{Cartier66a}.
The respective maps are
\begin{gather}
 \map{p}(s,x,y)= (\{x\},\{y\}), \qquad
 \map{s}(u,v)= (s,u,v), \nonumber\\
 \map{r}(s,x,y)= (s-\{x\}[y],[x],[y]),\label{eq:lattice-map-s-r}
\end{gather}
where $[x]$ is the integer part of $x$ and $\{x\}=x-[x]$ is the fractional part. Using them we express on $\Space{T}{2}$ as follows
\begin{align*}
 [\uir{m}{d}(s,x,y)f](u,v)&={\rm e}^{2\pi m\rmi (s+x\{v-y\}+u[v-y])} f(\{u-x\},\{v-y\})\nonumber\\
 &={\rm e}^{2\pi m\rmi (s+x(v-y)+(u-x)[v-y])} f(\{u-x\},\{v-y\}).
 \end{align*}
\end{enumerate}

For a unitary irreducible representation $\uir{}{}$ of $\Space[p]{H}{1}$, its restriction to the centre $Z$ is a character $\uir{}{}(s,0,0)= \rme^{2\pi \rmi \myhbar s}$, where the real parameter $\myhbar$ is known as the Planck constant in quantum mechanical contexts (see \cite[Section~2.2]{Kirillov94a}, \cite[Section~2.4.1]{Kirillov04a} and \cite[Section~1.3]{Folland89}) and provides a~mathematical tool to describe a~semi-classical limit. See also a meaningful assigning of physical units (see~\cite[Convention~2.1]{Kisil02e} and \cite[Remark~3.7]{AlmalkiKisil18a}). The theorem of {Stone--von Neumann}~\cite[Section~1.5]{Folland89} states that any two infinite-dimensional
irreducible strongly continuous unitary representations of~$\Space[p]{H}{1}$ with the same Planck constant are unitary equivalent.
 In this paper, motivated by the physical framework and for the sake of simplicity, we only consider the positive Planck constant~$\myhbar>0$.

\section[The covariant transform on H\_p\textasciicircum{}1]{The covariant transform on $\boldsymbol{\Space[p]{H}{1}}$}\label{Covariant Transform}

\subsection{An induced covariant transform}\label{sec:induc-covar-transf}

In this work we need an extended version of the covariant transform which covers the Banach space situation (see \cite[Section~2]{Kisil98a} and \cite{Kisil13a}). All representations in this paper are assumed to be strongly continuous.
\begin{defn}Let $\uir{}{}$ be a representation of a group in a space $V$, and $F$ be an operator from~$V$ to a~space~$U$. We define a \emph{covariant transform} $\oper{W}_F$ from $V$ to the space $L(G,U)$ of a $U$-valued function on $G$ by the formula
	\begin{equation}\label{gen-cov}
		\oper{W}_F\colon \ v\mapsto\tilde v=F\big(\uir{}{}\big(g^{-1}\big)v\big),\qquad v\in V,\quad g\in G.
	\end{equation}
\end{defn}
The fundamental property of the covariant transform $\oper{W}_F$~\eqref{gen-cov} is that $\oper{W}_F$ intertwines the representation $\uir{}{}$ and the \emph{left regular action} $\Lambda$ of $G$:{\samepage
\begin{equation} \label{eq:left-shift-itertwine}
 \oper{W}_F\circ \uir{}{}(g) =\Lambda(g)\circ \oper{W}_F \qquad\text{for all} \quad g\in G,
\end{equation}
where $\Lambda(g)\colon f(g') \mapsto f\big(g^{-1}g'\big)$, $g,g'\in G$.}

For the Gilmore--Perelomov coherent states (see \cite{Perelomov86} and \cite[Section~7.1.2]{AliAntGaz14a}), it is enough to have the covariant transform values on a homogeneous space rather than the entire group. We name it \emph{the induced covariant transform}~\cite{Kisil13a} due to its connections with induced representations (cf.~\eqref{intertwinging-irre-repr}). More specifically, for the Heisenberg group it is defined as follows.
Let $\uir{}{}$ be an irreducible unitary representation of $\Space[p]{H}{1}$ on a Hilbert space $\FSpace{H}{}$, and~$H$ be a closed subgroup of~$\Space[p]{H}{1}$. Let $X=\Space[p]{H}{1}/H$ be a homogeneous space. Let $\phi_0\in \FSpace{H}{}$ be a \emph{fiducial vector}, that is,
\begin{equation} \label{joint eigenvector}
		\uir{}{}(h)\phi_0=\chi(h) \phi_0, \qquad\text{for all} \quad h\in H,
	\end{equation}
for some character $\chi$ of $H$.
 The induced covariant transform $\oper{W}_{\phi_0}^{\uir{}{}}$ is a map from the Hilbert space $\FSpace{H}{}$ to a space $W(X)$ of functions on $X=\Space[p]{H}{1}/H$ given as follows
 \begin{equation} \label{eq:induce-wavelet-transform-1}
 	\oper{W}_{\phi_0}^{\uir{}{}}\colon \ f \mapsto \tilde{f}(\map{s}(x)) = \scalar{f}{\uir{}{}(\map{s}(x))\phi_0},\qquad x\in X,
\end{equation}
where $\map{s}\colon X\mapsto \Space[p]{H}{1}$ is a
Borel section (the right inverse of the natural projection $\map{p}\colon \Space[p]{H}{1}\rightarrow
\Space[p]{H}{1}/H$).

Note that, in the definition~\eqref{eq:induce-wavelet-transform-1} we use a linear functional $\phi_0\in V'$ as a special case of the operator $F\colon V\rightarrow \Space{C}{}$ in~\eqref{gen-cov}. Then, an adjusted notation of the covariant transform $\oper{W}_F$ will be~$\oper{W}_{\phi_0}^{\uir{}{}}$.

The main algebraic property of the induced covariant transform~\eqref{eq:induce-wavelet-transform-1} is that it
intertwines $\uir{}{}$ on $\FSpace{H}{}$ with a representation $\uir{}{\chi}$ on $W(X)$ induced by the character~$\chi$ of the subgroup~$H$. That is,
\begin{equation}\label{intertwinging-irre-repr}
	\uir{}{\chi}\circ \oper{W}_{\phi_0}^{\uir{}{}}=\oper{W}_{\phi_0}^{\uir{}{}}\circ \uir{}{}.
	\end{equation}
Alternatively, this can be observed by the fact that any function of the image of the induced covariant transform~\eqref{eq:induce-wavelet-transform-1} has the $H$-covariance property
\begin{equation}\label{eq:prove-H-covarianc}
	\tilde f(gh)=\bar \chi(h)\tilde f(g).
\end{equation}
The main analytic property of the induced covariant transform is formulated in terms of \emph{the matrix coefficient}:
 \begin{equation} \label{matrix coefficient}
 	\oper{W}(f,\phi)(x,y)=\scalar{f}{\rho (0,x,y)\phi}\, \qquad \text{for} \quad f, \phi\in \FSpace{H}{},
 \end{equation}
which is a continuous linear map $	\FSpace{H}{} \times\FSpace{H}{} \rightarrow \FSpace{L}{2}\big(\Space{R}{2}\big)$.
Moreover, the map~\eqref{matrix coefficient} is \emph{sesqui-unitary}, that is for all $f_1,\phi_1,f_2,\phi_2\in\FSpace{H}{}$,
 \begin{equation}\label{generaL-sesqui-unitary}
 \scalar{\oper{W}(f_1,\phi_1)}{\oper{W}(f_2,\phi_2)}_{\FSpace{L}{2}(\Space{R}{2}) }=\scalar{f_1}{f_2}_{\FSpace{H}{}}\overline{\scalar{\phi_1}{\phi_2}}_{\FSpace{H}{}}.
 \end{equation}

 \begin{Example}[the inverse Fourier transform] \label{ex:covariant-Fourier}
 To find the covariant transform which intertwines two forms of the Schr\"odinger representations~\eqref{eq:H1-schroedinger-rep} and~\eqref{eq:H1-schroedinger-rep-alt}, we shall take the fiducial vector which would be the eigenfunction for all representations
$\uir{}{\myhbar}(s,x,0)$~\eqref{eq:H1-schroedinger-rep}. That is,
 \begin{displaymath}
 [\uir{}{\myhbar}(s,x,0)f](t)
 = \rme^{2\pi \rmi \myhbar s}\, f(t-x)
 = \rme^{2\pi \rmi \myhbar (s)} f(t).
 \end{displaymath} Of course the only solution is the constant function $f(t)\equiv c$. Then the corresponding induced covariant transform based on the maps~\eqref{eq:maps-p-s-r-H_y-prime} is
 \begin{equation*} 
 [\oper{W} f](y) = c \int_{\Space{R}{}} f(t) \rme^{2\pi \rmi \myhbar y t} \,\rmd t.
 \end{equation*}
 Of course, this transformation is the (inverse) Fourier transform which intertwines the representations~\eqref{eq:H1-schroedinger-rep} and~\eqref{eq:H1-schroedinger-rep-alt}. The significance of this intertwining property for the harmonic analysis is revealed in the work of Howe~\cite{Howe80a}.
 \end{Example}

Although many functions can be taken as fiducial vectors, some of them turn out to be much more preferable. The origin of their advantages is revealed by the following observation~\cite[Section~5]{Kisil11c}. Let $G$ be a Lie group and $\uir{}{}$ be its representation in a Hilbert
space~$\FSpace{H}{}$. Let $[\oper{W}_\phi f](g)=\scalar{f}{\uir{}{}(g)\phi}$ be the
covariant transform defined by a fiducial vector $\phi\in \FSpace{H}{}$. Then, the covariant transform
intertwines \emph{right shifts} $R(g)\colon f(g') \mapsto f(g'g)$ on the group~$G$ with the
associated action $\uir{}{}$ on fiducial vectors
\begin{equation*} 
 R(g) \circ \oper{W}_\phi= \oper{W}_{\uir{}{}(g)\phi}.
\end{equation*}
There are many interesting applications of this simple observation~\cite{AlmalkiKisil18a,AlmalkiKisil19a, Kisil11c,Kisil12d,Kisil17a}, in particular,
\begin{prop}[\cite{Kisil11c,Kisil12d}] \label{co:cauchy-riemann-integ}
Let $G$ be a Lie group with a Lie algebra $\algebra{g}$ and $\uir{}{}$ be a representation of $G$ on a Hilbert space $\FSpace{L}{2}(\Space{R}{n})$. We denote the derived representation of $\uir{}{}$ by $\rmd\uir{X}{}$. Let $\phi$ be a~fiducial vector in the Schwartz space $\oper{S}(\Space{R}{n})$ such that
 $\big(\sum_{j=1}^n a_j\,\rmd\uir{X_j}{}\big)\phi=0$, for some $a_j\in\Space{C}{}$. Then, the image of the covariant transform consists of functions $f$ such that:
 \begin{equation*} 
 \Bigg(\sum_{j=1}^n\bar{a}_j \linv{X_j}\Bigg) \tilde f =0,
 \end{equation*}
 where $\linv{X}$ denotes the \emph{Lie derivative}~-- the derivation of the right regular representation~$R$ of~$G$:
 \begin{equation*} 
 \linv{X} = \left. \frac{\rmd R\big(\rme^{tX}\big)}{\rmd t} \right|_{t=0} , \qquad \text{for} \quad X\in\algebra{g}.
 \end{equation*}
\end{prop}

An induced covariant transform from a representation $\uir{}{}$ on a Banach space $\FSpace{B}{}$ can be defined in a similar fashion. Let $H$ be a closed subgroup of $\Space[p]{H}{1}$ and $X=\Space[p]{H}{1}/H$ be a homogeneous space.	 Consider a continuous section $\map{s}\colon \Space[p]{H}{1}/H\rightarrow \Space[p]{H}{1}$, which is the right inverse of the natural projection $\map{p}\colon \Space[p]{H}{1}\rightarrow
\Space[p]{H}{1}/H$. Denote by $B^*$ {the dual space of the Banach space} $B$ and $\uir{}{}^*$ the {adjoint operator to} $\uir{}{}$. Let $l_0\in B^*$ {be a~non-zero test function such that}
 \begin{equation*}
		\uir{}{}^*(h)l_0=\bar\chi(h) l_0, \qquad\text{for all} \quad h\in H,
	\end{equation*}
for some character $\chi$ of $H$.
{The induced covariant transform} $\oper{W}_{l_0}$ is defined as (see \cite[Section~2]{Kisil98a} and \cite{Kisil13a})
 	\begin{equation} 	\label{wavelet_B}
 	\tilde { v}(x) =[\oper{W}_{l_0}v](x)= \big\langle \rho\big(\map{s}(x)^{-1}\big)v,l_0\big\rangle=\scalar{v}{\uir{}{}^*(\map{s}(x))l_0}.
 	\end{equation}
Similarly to the Hilbert space case~\eqref{intertwinging-irre-repr}, the induced covariant transform~\eqref{wavelet_B}
intertwines $\uir{}{}$ on~$B$ {with a representation} $\uir{}{\chi}$ on~$W(X)$ {induced by the character} $\chi$ of the subgroup $H$~\cite[Section~2, Proposition~2.6]{Kisil98a}.

Thus the induced covariant transform can be used as
\begin{itemize}\itemsep=0pt
\item reproducing formulae on representation spaces, and
\item intertwining operators between different realisations of equivalent representations.
\end{itemize}
Both applications will be illustrated below.

\subsection{The (pre-)Fock--Segal--Bargmann transform}\label{Pre-Fock-Segal-Bargmann Transform}
We look for an induced covariant transform $ \oper{W}_\phi^{\uir{}{\myhbar}}\colon \FSpace{L}{2}(\Space{R}{}) \rightarrow \FSpace{L}{2}\big(\Space{R}{2}\big)$, which intertwines the Schr\"odinger representation~\eqref{eq:H1-schroedinger-rep} and the left quasi-regular representation restricted to an irreducible component of $\FSpace{L}{2}\big(\Space{R}{2}\big)$.
 In fact,
for the character
$\chi_\myhbar(s,0,0)= \rme^{2\pi\rmi\myhbar s}$ of the centre $Z=\big\{(s,0,0)\in\Space[p]{H}{1}\colon
s\in\Space{R}{}\big\}$, any vector
$\phi\in\FSpace{L}{2}(\Space{R}{})$ satisfies the following specialisation of~\eqref{joint eigenvector}
\begin{equation*}
	\uir{}{\myhbar}(s,0,0)
\phi=\chi_\myhbar(s,0,0) \phi,\qquad \text{for all} \quad (s,0,0)\in Z.
\end{equation*}
Let $\map{s}\colon \Space[p]{H}{1}/Z\rightarrow \Space[p]{H}{1}\colon (x,y)\mapsto (0,x,y)$ be a continuous section~\cite[Section~5.1.2]{Alamer19a}.
Thus, for all $f\in\FSpace{L}{2}(\Space{R}{})$, the induced covariant transform $\oper{W}_{\phi}^{\uir{}{\myhbar}}$ for any fiducial vector $\phi\in\FSpace{L}{2}(\Space{R}{})$ is
 \begin{equation*} 
 \big[\oper{W}_{\phi}^{\uir{}{\myhbar}}(f)\big](x,y)=\int_{\Space{R}{}} f(t) \rme^{2\pi\rmi\myhbar ty} \phi(t-x)\,\rmd t.
 \end{equation*}
The main properties of $\oper{W}_\phi^{\uir{}{\myhbar}}$
 follow from the general properties of the covariant transform.
\begin{rem}\label{properties-pre-transform}
 Let $\phi\in\FSpace{L}{2}(\Space{R}{})$ be a fiducial vector such that $\norm{\phi}=1$.
 The covariant transform $\oper{W}_\phi^{\uir{}{\myhbar}}\colon \FSpace{L}{2}(\Space{R}{}) \rightarrow \FSpace{L}{2}\big(\Space{R}{2}\big)$ is
 a unitary intertwining operator between the Schr\"odinger representation~$\uir{}{\myhbar}$ on~$\FSpace{L}{2}(\Space{R}{})$ and the left quasi-regular representation $\Lambda_{\myhbar}$ restricted on the image space
 \begin{equation*}
 F_\phi\big(\Space{R}{2}\big):=\big\{\oper{W}_{\phi}^{\uir{}{\myhbar }}(f)\colon f\in \FSpace{L}{2}(\Space{R}{})\big\}.
 \end{equation*}
 In particular, $\Lambda_{\myhbar}$ is an irreducible representation on $F_\phi\big(\Space{R}{2}\big)$.
\end{rem}
So far all fiducial vectors seem to be equally suitable, yet its is common to give the strong preference to a vacuum vector of the Schr\"odinger representation~-- the Gaussian
\begin{equation} \label{eq:Gaussina-defn}
 \phi_{\myhbar\kappa}(t)=2^{1/4}\rme^{-\frac{\pi\myhbar}{\kappa} t^2} .
\end{equation}
This preference is explained by Proposition~\ref{co:cauchy-riemann-integ}: since the Gaussian is a null-solution
 to the annihilation operator
 \begin{equation}\label{annihilation-schrodinger}
 	a^-_{\uir{}{\myhbar}}=\rmd\uir{\kappa X-\rmi Y}{\myhbar}=-2\pi\myhbar t-\kappa\partial_t ,
 \end{equation}
 the image space of the induced covariant transform $\oper{W}_{\phi_{\myhbar\kappa}}^{\uir{}{\myhbar}}$~\eqref{eq:induce-wavelet-transform-1} is annihilated by the Lie derivative~$\linv{\kappa X+\rmi Y}$ (the derived representation from the right regular action $R_\myhbar$). This allows to give an intrinsic characterisation of the image space.

 Explicitly, $\oper{W}_{\phi_{\myhbar\kappa}}^{\uir{}{\myhbar}}$ is given by
 \begin{align}
\tilde{f}(x,y)&:=\big[\oper{W}_{\phi_{\myhbar\kappa}}^{\uir{}{\myhbar}}f\big](x,y)=\scalar{f}{\uir{}{\myhbar}(0,x,y)\phi_{\myhbar\kappa}}
 \nonumber\\
 & =\left(\frac{\myhbar}{\kappa}\right)^{1/2}2^{1/4}\int_{\Space{R}{}} f(t) \rme^{2\pi\rmi\myhbar ty} \rme^{-\frac{\pi \myhbar}{\kappa} (t-x)^2}\,\rmd t,\label{eq:prefsb-transform}
 \end{align}
where the measure is renormalised by the factor $({\frac{\myhbar}{\kappa}})^{1/2}$.
For reasons explained here, we call it the pre-FSB transform (see \cite[Section~~4.2]{Neretin11a} and \cite[Section~1.6]{Folland89}) from $\FSpace{L}{2}(\Space{R}{})$ into the \emph{pre-FSB space} $\FSpace{F}{\phi_{\myhbar\kappa}}\big(\Space{R}{2}\big)$~\eqref{pre-FSB}.
The image space $\FSpace{F}{\phi_{\myhbar\kappa}}\big(\Space{R}{2}\big)$ is a subspace of square-integrable functions on~$\Space{R}{2}$. The left quasi-regular representation $\Lambda_{\myhbar}$ restricted on the pre-FSB space $\FSpace{F}{\phi_{\myhbar\kappa}}\big(\Space{R}{2}\big)$ is called the \emph{pre-FSB representation}.
The prefix ``pre-'' is removed by a unitary
operator~-- the \emph{peeling}, which will produce the \emph{FSB space} of analytic functions on $\Space{C}{}$ in Section~\ref{exm:peeling-rep-center-FSB}.

\subsection{The Zak transform}\label{exp:calculating-Zak}
 In this subsection, we derive the Zak transform (see \cite[Chapter~9]{Berndt07a}, \cite{Cartier66a} and \cite[Chapter~4]{Folland89}) as a particular case of the covariant transform. More specifically, we look for the induced covariant transform
 \begin{equation*}
 \oper{W}_{l_0}^{\uir{}{\myhbar }}\colon \ \FSpace{L}{2}(\Space{R}{})\rightarrow \FSpace{L}{2}\big(\Space{T}{2}\big),
 \end{equation*}
 which intertwines the Schr\"odinger representation~\eqref{eq:H1-schroedinger-rep} and the lattice representation~\eqref{induced-smooth-rep-from-H_d}.
For an integer $m$, let $\chi_m(s,n.k)=\rme^{2 \pi \rmi m s}$ be the character of the subgroup $H_d=\{(s,n,k)\colon s\in \Space{R}{}$, $n,k\in\Space{Z}{}\}$~\eqref{non-commuta} of $\Space[p]{H}{1}$.
A required fiducial vector $l_0$ shall satisfy~\eqref{joint eigenvector}, which in this setup becomes
\begin{equation}\label{satifying-con-zak}
 \rme^{2\pi \rmi \myhbar s} \rme^{-2\pi \rmi \myhbar t k} l_0(t-n)=\rme^{2\pi \rmi m s} l_0(t), \qquad \text{for all} \quad (s,n,k)\in H_d.
\end{equation}
 The left- and the right-hand sides of~\eqref{satifying-con-zak} are equal if and only if
 \begin{enumerate}\itemsep=0pt
 \item[1)] $\myhbar=m$ (from considering $(1,0,0)\in H_d$),
 \item[2)] the function $l_0$ is a periodic function (with the period $1$) (from considering $(0,1,0)\in H_d$), and
 \item[3)] $\mathrm{supp}(l_0)\subseteq\Space{Z}{}$ (from considering all $(0,0,k)\in H_d$).
 \end{enumerate}
 Of course, the last two conditions imply that $\mathrm{supp}(l_0)=\Space{Z}{}$.
The simplest non-zero vector $l_0$ satisfying~\eqref{satifying-con-zak} would be the
\emph{Dirac comb} distribution, that is,
\begin{equation} \label{eq:Dirac-comb}
 l_0(t)=\sum_{n\in\Space{Z}{}}\delta(t-n),
\end{equation}
which is a periodic distribution constructed from the \emph{Dirac delta} $\delta(t)$.

\begin{rem}{Let} $K$ {be a compact subset of} $\Space{R}{}$, and $\FSpace{C}{}(K)$ be the space of continuous functions on~$\Space{R}{}$ {supported in}~$K$, which is a Banach space equipped with the uniform norm. Let $\FSpace{C}{c}(\Space{R}{})$ {be the union of these Banach spaces, where} $\FSpace{C}{c}(\Space{R}{})$ {inherits a natural inductive limit topology}~\cite{Folland16a, Rudin91a}. We denote by $\FSpace[*]{C}{c}(\Space{R}{})$ {the space of all continuous functionals (\emph{pseudomeasures}) on} $\FSpace{C}{c}(\Space{R}{})$, which is
 the intersection of all duals of $\FSpace{C}{}(K)$. {As the Dirac comb} $l_0(t)$~\eqref{eq:Dirac-comb} is a finite measure in any compact set~$K$, thus $l_0\in \FSpace[*]{C}{c}(\Space{R}{})$ {is a pseudomeasure. Therefore, we
 consider the Schr\"odinger representation} $\uir{}{\myhbar}$ of $\Space[p]{H}{1}$ on $\FSpace{C}{c}(\Space{R}{})\subset \FSpace{L}{2}(\Space{R}{})$
{to be restricted to the Banach space} $\FSpace{C}{}(K)$, {for any compact subset} $K$ of $\Space{R}{}$.
\end{rem}

 \begin{thm}\label{th:Zak}
 	 Let $l_0$ be the fiducial vector defined by \eqref{eq:Dirac-comb}. For $f\in\FSpace{L}{2}(\Space{R}{})$, let
 \begin{equation*}
 [\oper{Z}f](u,v)=\rme^{2\pi \rmi m u v}\sum_{n\in\Space{Z}{}} f(u+n)\,\rme^{2\pi\rmi m n v}
\end{equation*}
be the so-called~{\rm \cite{Alamer19a}} \emph{co-Zak transform}.
Then, we have
\begin{equation*}
 \oper{Z}f=\oper{W}_{l_0}^{\uir{}{\myhbar }}(f).	
\end{equation*}
 \end{thm}
\begin{proof} Let $\map{s}\colon \Space{T}{2}\rightarrow \Space[p]{H}{1}\colon (u,v)\mapsto (0,u,v)$ be a continuous section (cf.~\cite[Section~5.1.3]{Alamer19a}). Let~$\FSpace{C}{c}(K)$ be the space of smooth functions that are compact support in $K$. For $f\in \FSpace{C}{c}(K)\subset\FSpace{L}{2}(\Space{R}{})$, we calculate the induced covariant transform as follows:
 \begin{align}
\nonumber
 \big[\oper{W}_{l_0}^{\uir{}{\myhbar }}(f)\big](u,v)
 &=\scalar{f}{\uir{}{\myhbar}(0,u,v )l_0}
\\
\nonumber
&=\int_{\Space{R}{}} f(t) \rme^{2\pi\rmi \myhbar t v } \bar l_0(t-u)\,\rmd t
\\
\nonumber
&=\int_{\Space{R}{}} f(t)\, \rme^{2\pi \rmi m t v}\sum_{n\in\Space{Z}{}}\delta(t-(u+n)) \rmd t,\qquad \myhbar=m
\\
\nonumber
&=\sum_{n\in\Space{Z}{}}\int_{\Space{R}{}} f(t) \rme^{2\pi \rmi m t v }\delta(t-(u+n))\,\rmd t
\\
\nonumber
&=\sum_{n\in\Space{Z}{}} f(u+n) \rme^{2\pi \rmi m v (u+n)}
 \\
\label{eq:covariant-Zak1}
 &=\rme^{2\pi \rmi m u v}\sum_{n\in\Space{Z}{}} f(u+n) \rme^{2\pi \rmi m v n}.
 \end{align}
 This is the Zak transform (see \cite[Chapter~9]{Berndt07a}, \cite{Cartier66a}, \cite[Chapter~4]{Folland89}) up to the factor $\rme^{2\pi \rmi m u v}$.
\end{proof}

\begin{rem}\label{extended_Zak}Since the co-Zak transform is defined on $\FSpace{C}{c}(\Space{R}{})$, {which is dense on} $\FSpace{L}{2}(\Space{R}{})$, the co-Zak transform can be extended to be defined on the entire $\FSpace{L}{2}(\Space{R}{})$.
\end{rem}

The general properties of the covariant transform $\oper{W}_{l_0}^{\uir{}{\myhbar }}$ yield corresponding properties of the Zak transform.
\begin{cor} \label{ZAk-proper}
 For $f\in\FSpace{L}{2}(\Space{R}{})$, let
 \begin{equation*}
 [\oper{Z}f](u,v)=\rme^{2\pi \rmi m u v}\sum_{n\in\Space{Z}{}} f(u+n) \rme^{2\pi\rmi m n v}
\end{equation*}
be the co-Zak transform~\eqref{eq:covariant-Zak1}.
Then, we have the following properties:
 \begin{enumerate}\itemsep=0pt
 \item[$1.$] 
 The operator $\oper{Z}\colon \FSpace{L}{2}(\Space{R}{})\rightarrow \FSpace{L}{2}\big(\Space{T}{2}\big)$ intertwines the Schr\"odinger representation $\uir{}{\myhbar}$ and the lattice representation $\uir{}{m}$. That is, $ \uir{}{m}\circ \oper{Z}=\oper{Z}\circ \uir{}{\myhbar}$, for $\myhbar=m$.

 \item[$2.$] 
 The operator $\oper{Z}\colon \FSpace{L}{2}(\Space{R}{})\rightarrow \FSpace{L}{2}\big(\Space{T}{2}\big)$ is unitary.

 \item[$3.$] 
 The image space of $[\oper{Z}f](u,v)$ consists of functions $\tilde f(u,v)$ that have the double-quasi-periodic property on $\Space{R}{2}$.
 \end{enumerate}
\end{cor}

\begin{proof}
1.~Since $\oper{W}_{l_0}^{\uir{}{\myhbar }}=\oper{Z}$, the intertwining property in~\eqref{eq:left-shift-itertwine}, for $\myhbar=m$, implies
 \begin{equation*}
 \uir{}{m}\circ \oper{W}_{l_0}^{\uir{}{\myhbar}}=\oper{W}_{l_0}^{\uir{}{\myhbar}}\circ \uir{}{\myhbar}.
 \end{equation*}

2.~Since $\oper{W}_{l_0}^{\uir{}{\myhbar }}$ is an intertwining operator between two irreducible representations of $\Space[p]{H}{1}$, by \emph{Schur's lemma}~\cite[Theorem~8.2.1]{Kirillov76}, the covariant transform $\oper{W}_{l_0}^{\uir{}{\myhbar }}$ is a bijection $\FSpace{L}{2}(\Space{R}{})\rightarrow \FSpace{L}{2}\big(\Space{T}{2}\big)$.
 ``When you have an intertwining operator between irreducible
representations, it has to be terrible not to be unitary''~\cite{Howe80a}, if a right factor is used. The scaling can be checked from the identity
\[
 \norm[\FSpace{L}{2}(\Space{R}{})]{\chi_{[0,1]}} = \big\|\oper{W}_{l_0}^{\uir{}{\myhbar }}\chi_{[0,1]}\big\|_{\FSpace{L}{2}(\Space{T}{2})}
\]
for the indicator function $\chi_{[0,1]}$ of the interval $[0,1]$ with $\oper{W}_{l_0}^{\uir{}{\myhbar }}\chi_{[0,1]}(u,v)=\rme^{2\pi \rmi m (u+1)v}$.

3.~The image space of the induced covariant transform $\tilde f= \oper{W}_{l_0}^{\uir{}{\myhbar }}(f)$ has the $H$-covariance property~\eqref{eq:prove-H-covarianc} $\tilde{f}(gh)= {\bar\chi}(h)\tilde{f}(g)$,
for all $ g \in G$ and $ h\in H$.
For the subgroup $H_d= \{(s, n, k)=(s ,n+\rmi k)\colon s\in\Space{R}{},\,n, k\in \Space{Z}{} \}$, it has exactly the double-quasi-periodic property~\eqref{eq: H_d-covariance property}:
 \begin{equation*}
 \tilde{f}(u+n ,v+k)=\rme^{2\pi m\rmi u k}\tilde{f}(u,v).\tag*{\qed}	
 \end{equation*}\renewcommand{\qed}{}
\end{proof}

\begin{Example}[the pre-theta function]
 It is natural to evaluate the covariant transform of the vacuum vector~-- Gaussian $\phi_{\myhbar\kappa}$~\eqref{eq:Gaussina-defn} with $\myhbar=m$:
 \begin{align}
 \Phi_{m\kappa}(u,v)
 &:= [{\oper{Z}}\phi_{ m \kappa}](u,v) = \rme^{\pi \kappa (3\omega^2-\bar{\omega}^2-2\omega \bar{\omega}) / (4m)} \sum_{n\in\Space{Z}{}}\rme^{-\frac{\pi m}{\kappa}n^2} \rme^{ 2\pi\rmi n \omega}\nonumber\\
 &=\rme^{\pi \kappa (3\omega^2-\bar{\omega}^2-2\omega\bar{\omega})/(4m)} \Theta_{m\kappa}\left( \omega, \frac{\rmi m }{\kappa }\right), \label{calculating-theta-vacuum}
 \end{align}
 where $\omega=m(v+\rmi u/\kappa )\in\Space{C}{}$ and $\Theta_{m\kappa}$ is the \emph{Jacobi theta function}. We will use the notation $\Phi_{m\kappa}(\omega ,\bar{\omega})=\Phi_{m\kappa}(u,v)$ as well.
 By Lemma~\ref{lem:vacuum-basic}(2) $\Phi_{m\kappa}$, is a vacuum of the lattice representation.
\end{Example}

\subsection{The (pre-)theta transform}\label{integration-transform-theta}
In the present subsection, we look for an intertwining operator
\begin{equation*}
\oper{W}_{\Phi}^{\uir{}{m}}\colon \ \FSpace{L}{2}\big(\Space{T}{2}\big)\rightarrow \FSpace{L}{2}\big(\Space{R}{2}\big)
\end{equation*}
 between the lattice representation and the left quasi-regular representation restricted to an irreducible component of $\FSpace{L}{2}\big(\Space{R}{2}\big)$.
Although the formula of the left quasi-regular representation~\eqref{eq:H1-pol-rep-plane}
looks very similar to the lattice representation's formula~\eqref{induced-smooth-rep-from-H_d}:
 they act on different spaces $\FSpace{L}{2}\big(\Space{R}{2}\big)$ and $\FSpace{L}{2}\big(\Space{T}{2}\big)$, respectively.

 Let $\chi_\myhbar(s,0,0)= \rme^{2\pi\rmi\myhbar s}$ be the character of the centre $Z\subset \Space[p]{H}{1}$.
As was already mentioned (Section~\ref{Pre-Fock-Segal-Bargmann Transform}), any vector
$\Phi\in\FSpace{L}{2}\big(\Space{T}{2}\big)$ satisfies version $\uir{}{m}(s,0,0)
\Phi=\chi_\myhbar(s,0,0) \Phi$ of~\eqref{joint eigenvector}, for all $(s,0,0)\in Z$ and $\myhbar=m$. Thus, the respective covariant transform $\oper{W}_{\Phi}$ intertwines the lattice and quasi-regular representations.

We may be more specific and request that $\oper{W}_{\Phi}$ map: $\FSpace{L}{2}\big(\Space{T}{2}\big)$ to the pre-FSB space $\FSpace{F}{2}\big(\Space{R}{2}\big)$.
As in the case of the pre-FSB transform, the vacuum $\Phi_{m\kappa}$~\eqref{calculating-theta-vacuum} shall be taken as the fiducial vector. Indeed, by Proposition~\ref{co:cauchy-riemann-integ}, the image space of $\oper{W}_{\Phi_{m\kappa}}^{\uir{}{m}}$ is annihilated by the right ladder $\rmi \kappa\linv{ X}+\linv{Y}$~\eqref{right-action-center-rep}.
\begin{thm}\label{theo_Thata}
	Let $\Phi_{m\kappa}(u,v)
 =\rme^{\pi \kappa (3\omega^2-\bar{\omega}^2-2\omega\bar{\omega})/(4m)} \Theta_{m\kappa}\big( \omega, \frac{\rmi m }{\kappa }\big)$ be a fiducial vector defined in \eqref{calculating-theta-vacuum}.
 Then, the induced covariant transform $\oper{W}_{\Phi_{m\kappa}}^{\uir{}{m}}\colon \FSpace{L}{2}\big(\Space{T}{2}\big)\rightarrow \FSpace{F}{2}\big(\Space{R}{2}\big)$ is as follows
 \begin{gather*}
 \big[\oper{W}_{\Phi_{m\kappa}}^{\uir{}{m}}(f)\big](x,y)
 =\int_{\Space{T}{2}} f(\omega )\, \rme^{\pi \kappa[2(\zeta-\bar{\zeta})[(\bar{\omega}-\bar{\zeta})+(\omega-\zeta)]+3 (\bar{\omega}-\bar{\zeta})^2-(\omega-\zeta)^2-2(\omega-\zeta)(\bar{\omega}-\bar{\zeta})]/(4 m)
}\\
\hphantom{\big[\oper{W}_{\Phi_{m\kappa}}^{\uir{}{m}}(f)\big](x,y)
 =\int_{\Space{T}{2}}}{}
\times
 \overline{\Theta_{m\kappa}\left(\omega -\zeta,\frac{\rmi m}{\kappa }\right)} \frac{\kappa\, \rmd \omega \wedge\rmd \bar{\omega}}{\rmi m }.
\end{gather*}
\end{thm}
\begin{proof}The induced covariant transform $\oper{W}_{\Phi_{m\kappa}}^{\uir{}{m}}\colon \FSpace{L}{2}\big(\Space{T}{2}\big)\rightarrow \FSpace{F}{2}\big(\Space{R}{2}\big)$ is calculated (up to normalisation) as follows
\begin{gather}
 \big[\oper{W}_{\Phi_{m\kappa}}^{\uir{}{m}}(f)\big](x,y)
 =\scalar{f}{\uir{}{m}(\map{s}(x,y) )\Phi_{m\kappa}}\nonumber \\
 \hphantom{\big[\oper{W}_{\Phi_{m\kappa}}^{\uir{}{m}}(f)\big](x,y)}{}
 =\int_{\Space{T}{2}} f(u,v) \rme^{-2\pi\rmi m x(v-y)} \bar{\Phi}_{m\kappa}(u-x,v-y)\,\rmd u\,\rmd v\nonumber\\
\hphantom{\big[\oper{W}_{\Phi_{m\kappa}}^{\uir{}{m}}(f)\big](x,y)}{}
= f(\omega ) \rme^{\pi \kappa[2(\zeta-\bar{\zeta})[(\bar{\omega}-\bar{\zeta})+(\omega-\zeta)]+3 (\bar{\omega}-\bar{\zeta})^2-(\omega-\zeta)^2-2(\omega-\zeta)(\bar{\omega}-\bar{\zeta})]/(4 m)
}\nonumber\\
\hphantom{\big[\oper{W}_{\Phi_{m\kappa}}^{\uir{}{m}}(f)\big](x,y)=}{}
 \times
 \overline{\Theta_{m\kappa}\left(\omega -\zeta,\frac{\rmi m}{\kappa }\right)} \frac{\kappa\, \rmd \omega \wedge\rmd \bar{\omega}}{\rmi m },\label{theta-transform-1}
 \end{gather}
where $\omega=m(v+\rmi u/\kappa )$ and $\zeta=m(y+\rmi x/\kappa )$.
\end{proof}

 We call $\oper{W}_{\Phi_{m\kappa}}^{\uir{}{m}}$ the \emph{pre-theta transform}.

 The general properties of the covariant transform $\oper{W}_{\Phi_{m\kappa}}^{\uir{}{m}}$ yield corresponding properties of the pre-theta transform.
\begin{cor}The pre-theta transform $\oper{W}_{\Phi_{m\kappa}}^{\uir{}{m}}\colon \FSpace{L}{2}\big(\Space{T}{2}\big)\rightarrow \FSpace{F}{2}\big(\Space{R}{2}\big)$~\eqref{theta-transform-1} is a unitary intertwining operator between the lattice representation $\uir{}{m}$ on $\FSpace{L}{2}\big(\Space{T}{2}\big)$ and the left quasi-regular representation $\Lambda_{\myhbar}$ restricted to $F_{\Phi}\big(\Space{R}{2}\big)$.
	\end{cor}	

\section[Peeling representations of H\_p\textasciicircum{}1 and analyticity]{Peeling representations of $\boldsymbol{\Space[p]{H}{1}}$ and analyticity}\label{sec: Peeling}

Induced representations are a common tool to construct representations of groups. However, a~representation prepared using a generic methodology may not be particularly suited for a~special situation. It often needs to be tuned to be enriched with useful features. In this section, we demonstrate such a simple tool which produces the required enhancements needed for various situations.

For example: since the annihilation operator $a^{-}$ provides a useful characterisation of an irreducible component of a representation, we are interested in expressing $a^{-}$ in the most transparent form. A map, called here \emph{peeling}, then simplifies the corresponding annihilation operator into a~linear combination of first-order derivatives only.
Therefore, the structure of the eigenvectors~$\phi_n$ (cf.\ \cite[Section~5.4]{Alamer19a} and \cite[Section~2.6]{Kirillov04a}) forms an orthonormal basis of the initial irreducible space~$\FSpace{H}{}$ becomes more transparent.

On the other hand, if a representation is reducible, its irreducible component can be characterised as the space of null-solutions for the Lie derivative (cf.~\eqref{pre-FSB}). In such cases, we seek to peel the irreducible component to a space of analytic functions. This allows one to use the power of complex variable theory to study the induced representations of~$\Space[p]{H}{1}$.

\begin{defn}A \emph{peeling} $\varepsilon_d$ is an invertible operator of multiplication defined by a function~$d(x)$ on~$X$:
\begin{equation*}
	\varepsilon_d\colon \ f(x)\mapsto {\rm e}^{d(x)} f(x),
\end{equation*}
The operator $\varepsilon_d$ is unitary for suitably related measures
 \begin{equation*}
 	\varepsilon_d\colon \ \FSpace{L}{2}(X,\rmd\mu(x))\rightarrow \FSpace{L}{2}( X, \rmd\nu(x))
 \end{equation*}	
 such that $\rmd\nu(x)={\rm e}^{- 2\operatorname{Re} d(x)}\rmd\mu(x)$.
 \end{defn}
We use such peeling operators to improve some properties of covariant transform related to specific representations.
In this paper, all considered peelings use smooth $d(x)$ on a
domain $X$ in a Euclidean space.
We will discuss the choice of~$d(x)$ for the pre-FSB, Schr\"odinger and lattice representations in Sections~\ref{exm:peeling-rep-center-FSB},~\ref{peeling shrodinger} and~\ref{sub: peeling-Lattice}, respectively.

\subsection{Peeling the (pre-)FSB representation}\label{exm:peeling-rep-center-FSB}
Let $\Lambda_{\myhbar}$ be the pre-FSB representation~\eqref{eq:H1-pol-rep-plane}, which acts irreducibly on the pre-FSB space~$\!\FSpace{F}{\!\phi_{\myhbar\kappa}}\!\!$~\eqref{pre-FSB}.
 Consider the variables $z, \bar{z}\in\Space{C}{}$, where $z=\sqrt{\frac{h}{2\kappa}}(x+\rmi\kappa y)$ and $h=2\pi\myhbar>0$. 	
 In this subsection, we peel the representation $\Lambda_{\myhbar}$ into the corresponding one $\tilde \Lambda_{\myhbar}$ that acts
on the FSB space of analytic functions. To perform this, we look for a peeling operator satisfying the following conditions:
\begin{enumerate}\itemsep=0pt
\item 
The peeling defined by $\rme^{d(z, \bar{z})}$ shall intertwine the right annihilation operator $\linv{\kappa X+\rmi Y} =2\pi\myhbar x+(\kappa \partial_{x}+\rmi \partial_{y})$ and the Cauchy--Riemann operators $\partial_{\bar z}=\kappa \partial_{x}+\rmi\partial_{y}$:
 \begin{equation} \label{intertwining with lie derivative}
	\partial_{\bar{z}} \rme^{d(z, \bar{z})}f(z,\bar {z})=\rme^{d(z, \bar{z})}\linv{\kappa X+\rmi Y}f(z, \bar{z}).
\end{equation}
A simple differential equation for~\eqref{intertwining with lie derivative} implies that
\begin{equation} \label{eq:FSB-peeling-first}
 d(z,\bar {z})=\tilde\psi(z) + \frac{1}{4}(z+\bar {z})^2,
\end{equation}
 where $\tilde\psi$ is an arbitrary smooth function of $z$ alone.
\item 
Let $a^-_{\Lambda_{\myhbar}}=\rmd\Lambda_{\myhbar}^{\kappa X-\rmi Y}=2\pi\rmi\myhbar\kappa y-(\kappa \partial_{x}-\rmi\partial_{y})$ be the left annihilation operator of $\Lambda_{\myhbar}$~\eqref{eq:H1-pol-rep-plane}. The same peeling shall intertwine $a^-_{\Lambda_{\myhbar}}$ with (a multiple of) the complex derivative $\partial_{z}=(\kappa \partial_{x}-\rmi\partial_{y})$. This fixes $\tilde\psi(z)=-\frac{1}{2}z^2-c$, $c\in\Space{C}{}$ in~\eqref{eq:FSB-peeling-first}. Thus, the peeling operator becomes
 \begin{gather} \label{unique-FSB-peeling}
 \varepsilon_d\cdot I= \rme^{d(z, \bar z)}\cdot I=\rme^{\frac{h}{4\kappa}(x^2+\kappa^2 y^2-2\rmi\kappa x y)-c }\cdot I
=\rme^{\frac{1}{4}(\bar{z}^2-z^2+2z\bar{z})-c}\cdot I.
 \end{gather}
\end{enumerate}
Let $c_0\in\Space{C}{}$ and $c_0\neq0$. There is a special vacuum of the representation $\Lambda_\myhbar$ annihilated by both~ $a^-_{\Lambda_\myhbar}$ and $\linv{\kappa X+\rmi Y}$, given as follows
\begin{equation*}
 \phi_0(z,\bar {z}) =c_0\rme^{-\frac{1}{2}\bar {z}^2+\frac{1}{4} (z-\bar {z})^2+c} =c_0 \rme^{\frac{1}{4}(z^2-\bar {z}^2-2z\bar {z})+c}.
\end{equation*}
The consequence of the conditions~(1) and~(2) is that the peeling maps the vacuum $\phi_{0}$, which is killed by both the left and right annihilation operators to the function identically equal to $c_0$, which is killed by both $\partial_{z}$ and $\partial_{\bar{z}}$. The peeled representation $\tilde{\Lambda}_{\myhbar}$ is
\begin{align}
 [\tilde{\Lambda}_{\myhbar}(s,z)F] (z')& := \rme^{\frac{1}{4}(\bar{z'}^2-z'^2+2z'\bar{z'})}\Lambda_{\myhbar}(s,z)\rme^{-\frac{1}{4}(\bar{z'}^2-z'^2+2z'\bar{z'})}F(z') \nonumber\\
 &=\rme^{h\rmi s
 +\frac{1}{4}(\bar z^2-z^2-2z\bar z)+\bar z z' }F(z'-z) ,\label{FSB-Rep}
 \end{align}
which is called the \emph{FSB representation}.
The composition of the peeling~\eqref{unique-FSB-peeling} with the covariant transform~\eqref{eq:prefsb-transform} is
 \begin{align}
 F(z)&=\left(\frac{\myhbar}{\kappa}\right)^{1/2}\rme^{\frac{h}{4\kappa}(x^2+\kappa^2 y^2-2\rmi\kappa xy)-c }\int_{\Space{R}{}} f(t) \rme^{2\pi\rmi\myhbar t y}\, \rme^{-\frac{\pi \myhbar}{\kappa} (t-x)^2}\,\rmd t\nonumber\\
 &=\left(\frac{\myhbar}{\kappa}\right)^{1/2}\int_{\Space{R}{}} f(t) \rme^{-\frac{h}{2\kappa} t^2+\sqrt{\frac{2h}{\kappa}} t z-\frac{1}{2}z^2}\,\rmd t,\label{FSB-transform}
	\end{align}
where $F(z)$ is an analytic function of $z=\sqrt{\frac{h}{2\kappa}}(x+\rmi \kappa y)$.
Indeed, by Proposition~\ref{co:cauchy-riemann-integ} and the intertwining property~\eqref{intertwining with lie derivative}, the function $F(z)$~\eqref{FSB-transform}
satisfies
\begin{equation*}
 \partial_{\bar{z}} F(z) = (\kappa\partial_x+\rmi\partial_y)F(x,y)=0 ,
\end{equation*}
which is essentially the Cauchy--Riemann equation.
The integral~\eqref{FSB-transform} is known as the \emph{FSB transform}.
The image $\FSpace[\myhbar]{F}{2}$ of the FSB transforms is called the FSB space.
 It is a closed subspace of
 \begin{equation*}
 \FSpace{L}{2}\big(\Space{R}{2},\rme^{-\frac{h}{2\kappa}(x^2+\kappa^2 y^2)+2 c}\rmd x\,\rmd y\big)=\FSpace{L}{2}\big(\Space{C}{},\rme^{-|z|^2+2 c}\rmd z\,\rmd\bar z\big)	
 \end{equation*}
consisting of the analytic functions. Note that often, only the values $\myhbar=1$, $\kappa=1$ and $c=0$ are used~\cite{Vasilevski99}.

Similar to the pre-FSB transform, we calculate the composition of the pre-theta transform~$\tilde f$ \eqref{theta-transform-1} and the peeling~\eqref{unique-FSB-peeling}, $h=2\pi m$, as follows
 \begin{align*}
 \tilde F(z)&=\rme^{\frac{h}{4\kappa}(x^2+\kappa^2 y^2-2\rmi\kappa xy)-c } \tilde f(x,y)\\
 &=\int_{\Space{T}{2}}f(u,v) \rme^{-\frac{\pi m}{\kappa}(u^2+2\rmi\kappa uv)-\frac{1}{2} z^2+2\sqrt{\frac{\pi m}{\kappa}} u z-c} \bar\Theta_{m\kappa}(-\rmi (\bar{z'}-\bar{z}),\rmi)\,\rmd u\,\rmd v ,\label{FSB-transform-theta}
 \end{align*}
where $z=\sqrt{\frac{h}{2\kappa}}(x+\rmi \kappa y)\in\Space{C}{}$ as before.

We call $\tilde F(z)$ the \emph{theta transform} of $f(u,v)$. 	
As with the FSB transform, the intertwining property~\eqref{intertwining with lie derivative} implies that the image of the theta transform consists of analytic functions, which can be found in many references \cite{Berndt07a, Folland89,Neretin11a}.

\subsection{Peeling the Schr\"odinger representation}\label{peeling shrodinger}
In this subsection, we are performing less common peeling the Schr\"odinger representa\-tion $\uir{}{\myhbar}$~\eqref{eq:H1-schroedinger-rep}, so that the corresponding annihilation operator will only be the derivative $\partial_t$. For a representation space realised as a function on a set, we have the following simple result.
 \begin{lem}\label{le:peeling-vacuum}
 Let the annihilation operator have the form $a^{-}= M+D$ on $\FSpace{L}{2}(X)$, where~$M$ is a~multiplication operator and $D$ is a~derivative satisfying the Leibniz rule. For the vacuum vector $\phi(x)$, such that $a^{-} \phi(x)= 0 $, the peeling $f(x) \mapsto \phi(x) f(x)$ intertwines the annihilation operator~$a^{-}$ with the derivative~$D$.
 \end{lem}

Recall the annihilation operator
$a^-_{\uir{}{\myhbar}}=\rmd\uir{\kappa X-\rmi Y}{\myhbar}=-2\pi\myhbar t-\kappa\partial_t$~\eqref{annihilation-schrodinger} for the Schr\"odinger representation.
By Lemma~\ref{le:peeling-vacuum}, a function $\varepsilon_d$ which intertwines $a^-_{\uir{}{\myhbar}}$ with the plain derivative~$\kappa \partial_t$ shall satisfy the vacuum condition: $a^-_{\uir{}{\myhbar}} \varepsilon_d^{-1} = 0$. Thus, the peeling defined by $\varepsilon_d(t)=c \rme^{\frac{\pi\myhbar t^2}{\kappa}}$, $c\in \Space{C}{}$ is unitary:
 \begin{equation*}
 	\varepsilon_d\colon \ \FSpace{L}{2}(\Space{R}{},\rmd t)\longrightarrow \FSpace{L}{2}\big(\Space{R}{},\rme^{-2\frac{\pi\myhbar}{\kappa} t^2}\rmd t\big).
 \end{equation*}
 The composition of the Schr\"odinger representation with the peeling acting on $\FSpace{L}{2}\big(\Space{R}{},\rme^{-2\frac{\pi\myhbar}{\kappa} t^2}\rmd t\big)$ is
 \begin{equation*}
 [\tilde\rho_{\myhbar}(s,x,y)F](t) :=
\rme^{\frac{\pi\myhbar}{\kappa} t^2}\rme^{2\pi \rmi \myhbar
 (s-t y)}\rme^{-\frac{\pi\myhbar}{\kappa} (t-x)^2}F(t-x)
 =\rme^{2\pi \rmi \myhbar s}\rme^{-\frac{\pi\myhbar}{\kappa} (x^2-2t(x-\rmi\kappa y))}F(t-x) .
 \end{equation*}
Consequently, the corresponding derived representations of the Lie algebra $\algebra{h}_1$ are
 \begin{gather*}
 	\rmd\tilde\rho_{\myhbar}^{ X}=2\pi\myhbar t-\partial_t\qquad \text{and} \qquad\rmd\tilde\rho_{\myhbar}^{Y}=-2\pi\rmi\myhbar t,
 \end{gather*}
and the annihilation operator $a^-_{\tilde\rho_{\myhbar}}$ is
\begin{equation*}
	a^-_{\tilde\rho_{\myhbar}}=\rmd\tilde\rho_{\myhbar}^{\kappa X-\rmi Y}=-\partial_t,
\end{equation*}
which annihilates the vacuum $\tilde\phi(t)=c$.
A notable consequence is that $\varepsilon_d$ transforms the Hermite functions $H_n(t)\rme^{-\frac{\pi\myhbar}{\kappa} t^2}\in\FSpace{L}{2}(\Space{R}{})$ to the corresponding \emph{Hermite polynomials} $H_n(t) \in \FSpace{L}{2}\big(\Space{R}{},\rme^{-2\frac{\pi\myhbar}{\kappa} t^2}\rmd t\big)$.

\subsection{Peeling the lattice representation}\label{sub: peeling-Lattice}

The purpose of peeling the lattice representation $\uir{}{m}$~\eqref{induced-smooth-rep-from-H_d} seems to be difficult to formulate while staying within the framework of the Heisenberg group itself. However, the lattice representation is better understood if it is extended to the Schr\"odinger group (aka the Jacobi group~\cite[Chapter~9]{Berndt07a})~-- the semi-direct product of $\Space[p]{H}{1}$ with the group $\SL$ acting on $\Space[p]{H}{1}$ through symplectic automorphism~\cite[Section~1.2]{Folland89}. The representation of the Schr\"odinger group can then be peeled in a way that its vacuum will be a null-solution of the heat equation.

Since an accurate treatment of the Schr\"odinger group is beyond the scoop of the present paper, we simply provide the form of the required peeling, cf.~\eqref{calculating-theta-vacuum}:
\begin{equation*}
 \varepsilon_d\cdot I = \rme^{-\pi\kappa (3\omega^2-\bar\omega^2-2\omega\bar{\omega})/4m}\cdot I, \qquad
 \text{where} \quad \omega = \frac{m}{\kappa}(\kappa v+\rmi u).
\end{equation*}
Then, $\varepsilon_d $ is a unitary operator:
 \begin{equation*}
 \varepsilon_d\colon \ \FSpace{L}{2}\big(\Space{T}{2}, \rmd\omega \rmd\bar{\omega}\big) \longrightarrow \FSpace{\hat L}{2}\big(\Space{T}{2}, \rme^{\pi\kappa(\omega- \bar{\omega})^2/2m} \rmd\omega \rmd\bar{\omega}\big).
 \end{equation*}
Note that $\pi\kappa(\omega -\bar{\omega})^2/2m=-2\operatorname{Re}\big[{-}\pi\kappa \big(3\omega^2-\bar\omega^2-2\omega\bar{\omega}\big)/4m\big]= -2\pi m u^2/\kappa$.
The corresponding irreducible lattice representation $\tilde\rho_m$ is
\begin{equation*}
 [\tilde\rho_m (s,\omega)F](\omega',\bar{\omega'})
 =\rme^{2\pi m\rmi s} \rme^{\pi\kappa(\frac{1}{4}(\omega -\bar{\omega})^2+\omega(\bar{\omega'}-\omega'))/m}F(\omega'-\omega, \bar{\omega'}-\bar{\omega}) .
\end{equation*}
Therefore, the corresponding annihilation operator is simply
\begin{equation*}
	a^-_{{\tilde\rho}_{m}}= \rmd{\tilde\rho_m}^{\kappa X-\rmi Y}=2\rmi m \partial_{\bar{\omega}},
\end{equation*}
 which annihilates the theta function $\tilde\Phi=:c \Theta_{m\kappa}\big(\omega,\frac{\rmi m}{\kappa}\big)$.
Therefore, the peeling maps the vacuum vector $\Phi_{m\kappa}$ from Section~\ref{integration-transform-theta} to the theta function:
 \begin{equation*}
 \varepsilon_d\colon \ \Phi(\omega,\bar \omega) \mapsto
 c\Theta_{m\kappa}\left(\omega,\frac{\rmi m}{\kappa}\right).
 \end{equation*}

\section[A contravariant transform on H\_p\textasciicircum{}1]{A contravariant transform on $\boldsymbol{\Space[p]{H}{1}}$}\label{Contravariant Transform}

The goal of this section is to introduce the \emph{contravariant transform} $\oper{M}_\psi$, which is the \emph{adjoint} of the covariant transform $\oper{W}_\phi$ (see \cite[Section~8.1]{AliAntGaz14a} and \cite[Section~2]{Kisil98a}).

\subsection{A Contravariant transform for induced representations}\label{sec:contr-transf-induc}

Let $H$ be a closed subgroup of $\Space[p]{H}{1}$ and $X=\Space[p]{H}{1}/H$ be the respective homogeneous space, which is a subset of Euclidean space with the respective Lebesgue measure.
 Let $\uir{}{}$ be a representation of the Heisenberg group $\Space[p]{H}{1}$ on the vector space $\FSpace{V}{}$.
 \begin{defn} \label{de:contravariant-tran}
 For a \emph{reconstruction vector} $\psi\in \FSpace{V}{}$, the \emph{contravariant transform} $\oper{M}_\psi$
 is a map $\oper{M}_\psi^{\uir{}{}}\colon \FSpace{L}{1}(X) \rightarrow \FSpace{V}{}$ given by
 \begin{equation*}
 \oper{M}_\psi^{\uir{}{}}\colon \ \tilde \nu \mapsto \oper{M}_\psi(\tilde \nu )
 =\int_{X}\tilde \nu (x) \uir{}{}(\map{s}(x)) \psi\,\rmd\mu(x),
 \end{equation*}
 where $\map{s}$ is a~Borel section from $X=\Space[p]{H}{1}/H$ to $\Space[p]{H}{1}$.
 \end{defn}

It is naturally to request that the contravariant transform $\oper{M}_\psi^{\uir{}{}}$ shall intertwine the induced representation $\uir{}{\chi}$~\eqref{eq:def-ind} and the representation~$\uir{}{}$:
\begin{equation} \label{eq:contravariant-intertwining}
 \oper{M}_\psi^{\uir{}{}} \circ \uir{}{\chi}(g) = \uir{}{}(g) \circ \oper{M}_\psi^{\uir{}{}}, \qquad \text{for all} \quad g \in \Space[p]{H}{1}.
\end{equation}
The left-hand side explicitly is
\begin{align}
 \nonumber
 [\oper{M}_\psi^{\uir{}{}} \circ \uir{}{\chi}(g) \tilde{\nu}]
 &= \int_{X} \bar{\chi}\big(\map{r}\big(g^{-1} * \map{s}(x)\big)\big)
 \tilde{\nu} \big(g^{-1}\cdot x\big) \uir{}{}(\map{s}(x)) \psi\,\rmd\mu(x)\\
 \nonumber
 &= \int_{X}
 \tilde{\nu} (y) \bar{\chi}(\map{r}(g * \map{s}(y))) \uir{}{} \big(g*\map{s}(y)*(\map{r}(g*\map{s}(y)))^{-1}\big) \psi\,\rmd\mu(x)\\
 \nonumber
 &= \int_{X}
 \tilde{\nu} (y) \bar{\chi}(\map{r}(g * \map{s}(y))) \uir{}{}\left(\map{s}(\map{p}(g*\map{s}(y)))\right) \psi\,\rmd\mu(x)\,,
\end{align}
where $y = g^{-1}\cdot x$ for $x$, $y\in X$ and $g\in G$. We used the following identities:
\begin{gather*}
 \map{s}(x) = g*\map{s}(y)*(\map{r}(g*\map{s}(y)))^{-1},\qquad
 \map{r}\big(g^{-1}*\map{s}(x)\big) = \map{r}(g*\map{s}(y)).
\end{gather*}
The right-hand side of~\eqref{eq:contravariant-intertwining} is
\[
\big[\uir{}{}(g) \circ \oper{M}_\psi^{\uir{}{}} \tilde{\nu}\big] = \int_{X}\tilde \nu (y)\uir{}{}(g*\map{s}(y)) \psi\,\rmd\mu(y) .
\]
If the intertwining property holds for every function $\tilde{\nu}$, then the following condition for the reconstructing vector $\psi$ is required:
\begin{equation} \label{eq:condition-intertwining}
 \bar{\chi}(\map{r}(g))\uir{}{}\left(\map{s}(\map{p}(g))\right) \psi= \uir{}{}(g) \psi, \qquad \text{for all} \quad g \in G.
\end{equation}

The abstract framework of the contravariant transform is well known for a unitary irreducible representation $\uir{}{}$ in a Hilbert space $\FSpace{V}{}$ (see for example~\cite[Section~8.1]{AliAntGaz14a}).

For suitable fiducial $\phi$ and reconstruction $\psi\in \FSpace{V}{}$ vectors, the contravariant transform $\oper{M}_\psi$ and the covariant transform $\oper{W}_\phi$~\eqref{eq:induce-wavelet-transform-1} are adjoints.

There are several extensions of the constructions for a strongly continuous representation $\uir{}{}$ in a Banach space $\FSpace{V}{}$ \cite{Albargi15a, FeichGroech89a,FeichGroech89b,Kisil98a,Kisil13a}. In this paper we omit generalities and use specialised techniques for $\Space[p]{H}{1}$ \cite{Albargi15a, Kisil98a,Kisil13a}.
We use $\oper{W}^*\colon B^* \rightarrow W^*(X)$ and $\oper{M}^*\colon W^*(X) \rightarrow B^*$ to denote the adjoint operators to $\oper{W}$ and $\oper{M}$, respectively. This results in the following identity:
\begin{equation*}
	\scalar{\oper{M}v}{\oper{M}^*l}_{W(X)}=\scalar{v}{l}_{B},\qquad v\in B, \quad l\in B^*.
\end{equation*}

The contravariant transform construction is particularly simple for maps acting from the pre-FSB space as the next two examples show.

\begin{Example}[the inverse of the (pre-)FSB transform]
For the Schr\"odinger representation $\uir{}{\myhbar}$~\eqref{eq:H1-schroedinger-rep}, the intertwining condition~\eqref{eq:condition-intertwining} is trivially satisfied by an arbitrary reconstructing vector $\psi$. From the sesqui-unitary property~\eqref{generaL-sesqui-unitary}, it follows that $\oper{M}_{\psi}\circ \oper{W}_\phi = I$ for vectors~$\psi$ and~$\phi$ such that $\scalar{\phi}{\psi}=1$. In particular, for the Gaussian $\psi_{\myhbar\kappa}(t) =2^{1/4} \rme^{-\frac{\pi\myhbar}{\kappa} t^2}$ as a reconstruction vector, the contravariant transform is
\begin{align*}
 \big[\oper{M}_\psi^{\uir{}{\myhbar}}f\big](t)&= 2^{1/4} \int_{\Space{R}{2}}f(x,y) \rme^{-2\pi \rmi \myhbar ty} \rme^{-\frac{\pi\myhbar}{\kappa} (t-x)^2}\,\rmd x\,\rmd y \\
 &=2^{1/4}\rme^{-\frac{\pi\myhbar}{\kappa}t^2} \int_{\Space{R}{2}}f(x,y) \rme^{-\frac{\pi\myhbar}{\kappa} (x^2-2t(x-\rmi \kappa y))}\,\rmd x\,\rmd y.
 \end{align*}	 	
This is known as the \emph{inverse of the pre-FSB transform}~\cite[Section~4.2]{Neretin11a}.
\end{Example}

\begin{Example}[the inverse of the (pre-)theta transform]
 A contravariant transform $\oper{M}_{\psi_{\myhbar\kappa}}^{\uir{}{m}}$: $\FSpace{L}{2}\big(\Space{R}{2}\big)\rightarrow \FSpace{L}{2}\big(\Space{T}{2}\big)$ is similar to the previous case: there are no restrictions from the intertwining condition~\eqref{eq:condition-intertwining}, and the only requirement for a reconstruction vector $\psi$ is $\scalar{\phi}{\psi}=1$.
In particular, if we set the reconstruction vector $\psi$ by the lattice representation's vacuum~\eqref{calculating-theta-vacuum}
 the integral transformation $\oper{M}_\psi^{\uir{}{m}}$ becomes:
\begin{gather}
 \oper{M}_{\psi_{\myhbar\kappa}}^{\uir{}{m}}(f) = \rme^{-\frac{\pi m}{\kappa}u^2-2\rmi\kappa u v} \int_{\Space{T}{2}}f(x,y)\,\rme^{-\frac{\pi m}{\kappa}(x^2-2 u (x- \rmi\kappa y))}\nonumber\\
\hphantom{\oper{M}_{\psi_{\myhbar\kappa}}^{\uir{}{m}}(f) =}{}
 \times \Theta_{m\kappa}\left(\sqrt{\frac{h}{2\kappa}}(\kappa (v-y)+\rmi( u-x),\rmi\right)\,\rmd x\,\rmd y.\label{inver-theta-tra}
\end{gather}	
Thus, we obtain the inverse operator of $\oper{W}_{\Phi_{\myhbar\kappa}}^{\uir{}{m}}$~\eqref{theta-transform-1}. We call the transformation $\oper{M}_{\psi_{\myhbar\kappa}}^{\uir{}{m}}$~\eqref{inver-theta-tra} the inverse of the pre-theta transform.
\end{Example}
\begin{Example}[the Fourier transform]
 We can look for a contravariant transform which intertwines the two forms~\eqref{eq:H1-schroedinger-rep} and~\eqref{eq:H1-schroedinger-rep-alt} and will be an inverse of the covariant transform (the Fourier transform) from Example~\ref{ex:covariant-Fourier}. Using maps~\eqref{eq:maps-p-s-r-H_y-prime} and representation~\eqref{eq:H1-schroedinger-rep}, we obtain the form of the compatibility condition~\eqref{eq:condition-intertwining}:
 \[
 \rme^{2\pi \rmi \myhbar s} \rme^{-2\pi \rmi \myhbar ty} \psi(t) = \rme^{2\pi \rmi \myhbar (s-ty)} \psi(t-x),
 \]
 which again delivers the solution $\psi(\lambda)\equiv 1$. The respective contravariant transform is, as expected, the Fourier transform:
 \[
 \big[\oper{M}_{1}^{\uir{}{\myhbar}} f\big](t) = \int_{\Space{R}{}} f(\lambda) \rme^{-2\pi \rmi \myhbar t \lambda}\, \rmd \lambda .
 \]
\end{Example}
A bit more care is required in the next case.

\subsection{The inverse of the Zak transform}\label{inver-zak-contravariant-schro}

In Section~\ref{exp:calculating-Zak}, we derived the co-Zak transform ${\oper{Z}}\colon \FSpace{L}{2}(\Space{R}{})\rightarrow \FSpace{L}{2}\big(\Space{T}{2}\big)$~\eqref{eq:covariant-Zak1} through the induced covariant transform $\oper{W}_{\phi_0}^{\uir{}{\myhbar}}$.
 Now, we calculate its inverse using the contravariant transform.
The intertwining property~\eqref{eq:condition-intertwining} for the map $\map{r}(s,x,y)=(s-\{x\}[y],[x],[y])$ from~\eqref{eq:lattice-map-s-r} requires:
\begin{gather}
 \nonumber
 \bar{\chi}(s-\{x\}[y],[x],[y]) \uir{}{}\big((s,x,y)*(s-\{x\}[y],[x],[y])^{-1}\big) \psi= \uir{}{}(s,x,y) \psi , \\
\nonumber
 \bar{\chi}(s-\{x\}[y],[x],[y]) \uir{}{}\big((s,x,y)*(-s+x[y],-[x],-[y])\big) \psi = \uir{}{}(s,x,y) \psi, \\
\nonumber
 \bar{\chi}(s-\{x\}[y],[x],[y]) \uir{}{}(0,\{x\},\{y\}) \psi
 = \uir{}{}(s,x,y) \psi, \\
 \nonumber
 \rme^{2\pi \rmi m (s-\{x\}[y])} \rme^{2\pi \rmi m (-t\{y\})} \psi(t-\{x\})= \rme^{2\pi \rmi m (s-ty)}\, \psi(t-x) , \\
 \label{left}
 \rme^{2\pi \rmi m ((t-\{x\})[y])} \psi(t-\{x\}) = \psi(t-x) .
\end{gather}
Choosing $y=1$ and any $0<x<1$, we conclude that $\psi_0$ is supported at $\{0\}$. Further analysis shows that the reconstruction vector $\psi_0$ satisfying the condition~\eqref{left} is the Dirac delta distribution $\delta(t)$ (up to a multiple).

 Let $\tilde x=(u,v)\in X=\Space{T}{2}=\Space[p]{H}{1}/H_d$, where $H_d$ is the non-commutative subgroup~\eqref{non-commuta}. For the section map
 $\map{s}(u,v) = (0,u,v)$~\eqref{eq:lattice-map-s-r} and $g\in\FSpace{L}{2}\big(\Space{T}{2}\big)$, the contravariant transform becomes:
 \begin{align}
 \oper{M}_{\psi_0}^{\uir{}{\myhbar }}\colon \ g \mapsto & \int_{\Space{T}{}}\int_{\Space{T}{}}g(u,v) \uir{}{\myhbar}(\map{s}(\tilde x) ) \psi_0(t)\,\rmd u\,\rmd v\nonumber\\
 &=\int_{0}^1\int_{0}^1g(u,v) \uir{}{\myhbar}(0,u,v) \psi_0(t)\,\rmd u\,\rmd v\nonumber\\
 &=\int_{0}^1\int_{0}^1g(u,v) \rme^{-2\pi\rmi m t v}\delta(t-u)\,\rmd u\,\rmd v\nonumber\\
 &=\int_{0}^1g(t,v) \rme^{-2\pi\rmi m t v}\,\rmd v\nonumber\\
 &= \int_0^1 \tilde{g}(t,v)\,\rmd v. \label{eq:invers-zak1}
 \end{align}
 Since $g(t,v)$ is contained in the space $ \FSpace{L}{2}\big(\Space{T}{2}\big)$ of square-integrable functions that are periodic in $t$ and quasi-periodic in~$v$, multiplying $g(t,v)$ by $\rme^{-2\pi\rmi m t v}$ produces a function that has the same double quasi-periodicity property of $g(t,v)$ but in the opposite way. In other words, $\tilde{g}(t,v)=g(t,v)\cdot\rme^{-2\pi\rmi m t v}\in\FSpace{\tilde L}{2}\big(\Space{T}{2}\big)$ is quasi-periodic in $t$ and periodic in $v$ and square-integrable. Moreover, since $t\in \Space{R}{}\approx [0,1]\times\Space{Z}{}$, then $t=x+n$, for some $x\in [0,1]$ and $n\in\Space{Z}{}$.
 	Therefore, for $t=x+n$,~\eqref{eq:invers-zak1} becomes
 	\begin{equation}\label{eq:invers-zak}
 		\int_{0}^1\tilde{g}(x+n,v)\,\rmd v =\int_{0}^1\tilde{g}(x,v) \rme^{-2\pi\rmi m n v}\,\rmd v =\big[\oper{M}_{\psi_0}^{\uir{}{\myhbar }} g\big](x)=\big[\oper{M}_{\psi_0}^{\uir{}{\myhbar }} \big(\rme^{2\pi\rmi m x v} \tilde g\big)\big](x).
 	\end{equation}
 Thus, $\oper{M}_{\psi_0}^{\uir{}{\myhbar }}$ is the inverse of the induced covariant transform $\oper{W}_{\phi_0}^{\uir{}{\myhbar }}$~\eqref{eq:covariant-Zak1} from $ \FSpace{L}{2}\big(\Space{T}{2}\big)$ into $ \FSpace{L}{2}(\Space{R}{})$.

 We now provide the standard definition of the inverse of the Zak transform.
Let $\FSpace{\tilde L}{2}\big(\Space{T}{2}\big)$ be a space of square-integrable functions $\tilde g(x,v)$ that are quasi-periodic in $x$ and periodic in~$v$. Let $\tilde g= \oper{\tilde Z}f$ be the Zak transform of $f\in \oper{S}(\Space{R}{})\subset \FSpace{L}{2}(\Space{R}{})$. The function $f$ can be reconstructed using the following formula:
 	\begin{align}
 \oper{\tilde Z}^{-1}\colon \ & \FSpace{\tilde L}{2}\big(\Space{T}{2}\big)\rightarrow \FSpace{L}{2}(\Space{R}{}),\nonumber\\
 & \big[\oper{\tilde Z}^{-1}\tilde g\big](x)=\int_{\Space{T}{}}\tilde g(x,v)\rme^{-2\pi\rmi m n v}\,\rmd v,\qquad
 		n\in \Space{Z}{},\quad m\in \Space{N}{}.\label{inversion of the Zak}
 	\end{align}

 \begin{defn}[{\cite[Section~8.1]{Neretin11a}}]\label{the inverse of the Zak transform}
The operator $\oper{\tilde Z}^{-1}$~\eqref{inversion of the Zak}	is called {\em the inverse of the Zak transform}.
\end{defn}
The computations of this subsection allows us to present~$\oper{\tilde Z}^{-1}$ as a contravariant transform.
\begin{thm}
 \label{th:inverse-Zak-contravar}
 Let $\tilde{g}(x,v)=g(x,v)\cdot\rme^{-2\pi\rmi m x v}$ such that $g\in\FSpace{L}{2}\big(\Space{T}{2}\big)$. The contravariant transform~\eqref{eq:invers-zak},
 	 \begin{equation*}
 	 	\big[\oper{M}_{\psi_0}^{\uir{}{\myhbar }} \big(\rme^{2\pi\rmi m x v} \tilde g\big)\big](x)=\int_{0}^1\tilde{g}(x,v) \rme^{-2\pi \rmi m n v}\,\rmd v= \big[\oper{\tilde Z}^{-1}\tilde g\big](x),
 \end{equation*}
 is the inverse of the Zak transform. For $g\in\FSpace{L}{2}\big(\Space{T}{2}\big)$, one can write the inverse of the co-Zak transform as $\oper{Z}^{-1}g=\oper{\tilde Z}^{-1} \rme^{-2\pi \rmi m u v}g $.
 \end{thm}

\section{Conclusion}

Our work in this paper illustrates the technique which allows to obtain co- and contra-variant transforms with desired properties. For the covariant transform, the fiducial vector needs to agree with Proposition~\ref{co:cauchy-riemann-integ}. An induced contravariant transform will be an intertwining operator with an induced representation if the reconstructing vector satisfies~\eqref{eq:condition-intertwining}. This approach is illustrated on various representations of the Heisenberg group, and produces an interpretation of the Zak transform and its inverse as induced co- and contra-variant transforms.
Furthermore, we used peeling operators to obtain the familiar representation spaces of analytic function space of analytic functions.

Of course, this approach is not limited to the illustrative example of the Heisenberg group and can be fruitfully applied in many other situations.

\subsection*{Acknowledgements}

We are grateful to anonymous referees for many useful comments and remarks. Their suggestions were used for the paper's improvements.

\pdfbookmark[1]{References}{ref}
\LastPageEnding


\begin{thebibliography}{99}
\footnotesize\itemsep=0pt

\bibitem{Alamer19a}
Al~Ameer A.A., Singularities of analytic functions and group representations,
 Ph.D.~Thesis, {U}niversity of Leeds, 2019, available at
 \url{https://etheses.whiterose.ac.uk/24776/}.

\bibitem{Albargi15a}
Albargi A.H.A., Covariant transforms on locally convex spaces, Ph.D.~Thesis,
 {U}niversity of Leeds, 2015, available at
 \url{https://etheses.whiterose.ac.uk/12308/}.

\bibitem{AliAntGaz14a}
Ali S.T., Antoine J.P., Gazeau J.P., Coherent states, wavelets, and their
 generalizations, 2nd~ed., \textit{Theoretical and Mathematical Physics}, \href{https://doi.org/10.1007/978-1-4614-8535-3}{Springer}, New
 York, 201,.

\bibitem{AlmalkiKisil18a}
Almalki F., Kisil V.V., Geometric dynamics of a harmonic oscillator, arbitrary
 minimal uncertainty states and the smallest step 3 nilpotent {L}ie group,
 \href{https://doi.org/10.1088/1751-8121/aaed4d}{\textit{J.~Phys.~A: Math. Theor.}} \textbf{52} (2019), 025301, 25~pages,
 \href{https://arxiv.org/abs/1805.01399}{arXiv:1805.01399}.

\bibitem{AlmalkiKisil19a}
Almalki F., Kisil V.V., Solving the {S}chr\"odinger equation by reduction to a
 first-order differential operator through a coherent states transform,
 \href{https://doi.org/10.1016/j.physleta.2020.126330}{\textit{Phys. Lett.~A}} \textbf{384} (2020), 126330, 7~pages,
 \href{https://arxiv.org/abs/1903.03554}{arXiv:1903.03554}.

\bibitem{ArefijamaalGhaani13a}
Arefijamaal A.A., Ghaani~Farashahi A., Zak transform for semidirect product of
 locally compact groups, \href{https://doi.org/10.1007/s13324-013-0057-6}{\textit{Anal. Math. Phys.}} \textbf{3} (2013),
 263--276, \href{https://arxiv.org/abs/1203.1509}{arXiv:1203.1509}.

\bibitem{BarbieriHernandezMayeli14a}
Barbieri D., Hern\'andez E., Mayeli A., Bracket map for the {H}eisenberg group
 and the characterization of cyclic subspaces, \href{https://doi.org/10.1016/j.acha.2013.12.002}{\textit{Appl. Comput. Harmon.
 Anal.}} \textbf{37} (2014), 218--234, \href{https://arxiv.org/abs/1303.2350}{arXiv:1303.2350}.

\bibitem{BarbieriHernandezPaternostro15a}
Barbieri D., Hern\'andez E., Paternostro V., The {Z}ak transform and the
 structure of spaces invariant by the action of an {LCA} group,
 \href{https://doi.org/10.1016/j.jfa.2015.06.009}{\textit{J.~Funct. Anal.}} \textbf{269} (2015), 1327--1358, \href{https://arxiv.org/abs/1410.7250}{arXiv:1410.7250}.

\bibitem{Berezin86}
Berezin F.A., The method of second quantization, 2nd ed., Nauka, Moscow, 1986.

\bibitem{Berndt07a}
Berndt R., Representations of linear groups: an introduction based on examples
 from physics and number theory, Vieweg, Wiesbaden, 2007.

\bibitem{Bohm93}
Bohm A., Quantum mechanics: foundations and applications, 3rd ed., \textit{Texts and
 Monographs in Physics}, Springer-Verlag, New York, 1993.

\bibitem{Cartier66a}
Cartier P., Quantum mechanical commutation relations and theta functions, in
 Algebraic {G}roups and {D}iscontinuous {S}ubgroups ({P}roc. {S}ympos. {P}ure
 {M}ath., {B}oulder, {C}olo., 1965), Amer. Math. Soc., Providence, R.I., 1966,
 361--383.

\bibitem{FeichGroech89a}
Feichtinger H.G., Gr\"ochenig K.H., Banach spaces related to integrable group
 representations and their atomic decompositions.~{I}, \href{https://doi.org/10.1016/0022-1236(89)90055-4}{\textit{J.~Funct.
 Anal.}} \textbf{86} (1989), 307--340.

\bibitem{FeichGroech89b}
Feichtinger H.G., Gr\"ochenig K.H., Banach spaces related to integrable group
 representations and their atomic decompositions.~{II}, \href{https://doi.org/10.1007/BF01308667}{\textit{Monatsh.
 Math.}} \textbf{108} (1989), 129--148.

\bibitem{Folland89}
Folland G.B., Harmonic analysis in phase space, \textit{Annals of Mathematics
 Studies}, Vol.~122, \href{https://doi.org/10.1515/9781400882427}{Princeton University Press}, Princeton, NJ, 1989.

\bibitem{Folland16a}
Folland G.B., A course in abstract harmonic analysis, 2nd ed., \textit{Textbooks in
 Mathematics}, CRC Press, Boca Raton, FL, 2016.

\bibitem{Grochenig01a}
Gr\"ochenig K., Foundations of time-frequency analysis, \textit{Applied and Numerical
 Harmonic Analysis}, \href{https://doi.org/10.1007/978-1-4612-0003-1}{Birkh\"auser Boston, Inc.}, Boston, MA, 2001.

\bibitem{Grossmann85a}
Grossmann A., Morlet J., Paul T., Transforms associated to square integrable
 group representations. {I}.~{G}eneral results, \href{https://doi.org/10.1063/1.526761}{\textit{J.~Math. Phys.}}
 \textbf{26} (1985), 2473--2479.

\bibitem{Grossmann86a}
Grossmann A., Morlet J., Paul T., Transforms associated to square integrable
 group representations. {II}.~{E}xamples, \textit{Ann. Inst. H.~Poincar\'e
 Phys. Th\'eor.} \textbf{45} (1986), 293--309.

\bibitem{HernandezLuthySikicSoriaWilson21a}
Hern\'andez E., Luthy P.M., \v{S}iki\'c H., Soria F., Wilson E.N., Spaces
 generated by orbits of unitary representations: a tribute to {G}uido {W}eiss,
 \href{https://doi.org/10.1007/s12220-020-00396-0}{\textit{J.~Geom. Anal.}} \textbf{31} (2021), 8735--8761.

\bibitem{HernandezSikisWeissWilson10a}
Hern\'andez E., \v{S}iki\'c H., Weiss G., Wilson E., Cyclic subspaces for
 unitary representations of {LCA} groups; generalized {Z}ak transform,
 \href{https://doi.org/10.4064/cm118-1-17}{\textit{Colloq. Math.}} \textbf{118} (2010), 313--332.

\bibitem{HernandezSikicWeissWilson11a}
Hern\'andez E., \v{S}iki\'c H., Weiss G.L., Wilson E.N., The {Z}ak
 transform(s), in Wavelets and Multiscale Analysis, \textit{Appl. Numer. Harmon.
 Anal.}, \href{https://doi.org/10.1007/978-0-8176-8095-4_8}{Birkh\"auser/Springer}, New York, 2011, 151--157.

\bibitem{Howe80a}
Howe R., On the role of the {H}eisenberg group in harmonic analysis,
 \href{https://doi.org/10.1090/S0273-0979-1980-14825-9}{\textit{Bull. Amer. Math. Soc. (N.S.)}} \textbf{3} (1980), 821--843.

\bibitem{Iverson15a}
Iverson J.W., Subspaces of {$L^2(G)$} invariant under translation by an abelian
 subgroup, \href{https://doi.org/10.1016/j.jfa.2015.03.020}{\textit{J.~Funct. Anal.}} \textbf{269} (2015), 865--913,
 \href{https://arxiv.org/abs/1411.1014}{arXiv:1411.1014}.

\bibitem{Iverson18a}
Iverson J.W., Frames generated by compact group actions, \href{https://doi.org/10.1090/tran/6954}{\textit{Trans. Amer.
 Math. Soc.}} \textbf{370} (2018), 509--551, \href{https://arxiv.org/abs/1509.06802}{arXiv:1509.06802}.

\bibitem{Iverson19a}
Iverson J.W., The {Z}ak transform and representations induced from characters
 of an abelian subgroup, in 13th International conference on Sampling Theory
 and Applications (SampTA) (July 08--12, 2019, Bordeaux, France), \href{https://doi.org/10.1109/SampTA45681.2019.9030887}{IEEE}, 2020,
 19451251, 4~pages, \href{https://arxiv.org/abs/1904.03527}{arXiv:1904.03527}.

\bibitem{KaniuthTaylor13a}
Kaniuth E., Taylor K.F., Induced representations of locally compact groups,
 \textit{Cambridge Tracts in Mathematics}, Vol.~197, \href{https://doi.org/10.1017/CBO9781139045391}{Cambridge University
 Press}, Cambridge, 2013.

\bibitem{Kirillov76}
Kirillov A.A., Elements of the theory of representations, \textit{Grundlehren
 der Mathematischen Wissenschaften}, Vol.~220, \href{https://doi.org/10.1007/978-3-642-66243-0}{Springer-Verlag}, Berlin~-- New
 York, 1976.

\bibitem{Kirillov94a}
Kirillov A.A., Introduction to the theory of representations and noncommutative
 harmonic analysis, in Representation Theory and Noncommutative Harmonic
 Analysis,~{I}, \textit{Encyclopaedia Math. Sci.}, Vol.~22, \href{https://doi.org/10.1007/978-3-662-03002-8_1}{Springer}, Berlin,
 1994, 1--156, 227--234.

\bibitem{Kirillov04a}
Kirillov A.A., Lectures on the orbit method, \textit{Graduate Studies in
 Mathematics}, Vol.~64, \href{https://doi.org/10.1090/gsm/064}{Amer. Math. Soc.}, Providence, RI, 2004.

\bibitem{Kisil98a}
Kisil V.V., Wavelets in {B}anach spaces, \href{https://doi.org/10.1023/A:1006394832290}{\textit{Acta Appl. Math.}} \textbf{59}
 (1999), 79--109, \href{https://arxiv.org/abs/math.FA/9807141}{arXiv:math.FA/9807141}.

\bibitem{Kisil02e}
Kisil V.V., {$p$}-mechanics as a physical theory: an introduction,
 \href{https://doi.org/10.1088/0305-4470/37/1/013}{\textit{J.~Phys.~A: Math. Gen.}} \textbf{37} (2004), 183--204,
 \href{https://arxiv.org/abs/quant-ph/0212101}{arXiv:quant-ph/0212101}.

\bibitem{Kisil11c}
Kisil V.V., Erlangen program at large: an overview, in Advances in applied
 analysis, \textit{Trends Math.}, \href{https://doi.org/10.1007/978-3-0348-0417-2_1}{Birkh\"auser/Springer Basel AG}, Basel, 2012, 1--94,
 \href{https://arxiv.org/abs/1106.1686}{arXiv:1106.1686}.

\bibitem{Kisil12b}
Kisil V.V., Operator covariant transform and local principle,
 \href{https://doi.org/10.1088/1751-8113/45/24/244022}{\textit{J.~Phys.~A: Math. Theor.}} \textbf{45} (2012), 244022, 10~pages,
 \href{https://arxiv.org/abs/1201.1749}{arXiv:1201.1749}.

\bibitem{Kisil13a}
Kisil V.V., Calculus of operators: covariant transform and relative
 convolutions, \href{https://doi.org/10.15352/bjma/1396640061}{\textit{Banach~J. Math. Anal.}} \textbf{8} (2014), 156--184,
 \href{https://arxiv.org/abs/1304.2792}{arXiv:1304.2792}.

\bibitem{Kisil12d}
Kisil V.V., The real and complex techniques in harmonic analysis from the point
 of view of covariant transform, \textit{Eurasian Math.~J.} \textbf{5} (2014),
 95--121, \href{https://arxiv.org/abs/1209.5072}{arXiv:1209.5072}.

\bibitem{Kisil17a}
Kisil V.V., Symmetry, geometry and quantization with hypercomplex numbers, in
 Geometry, Integrability and Quantization {XVIII}, Bulgar. Acad. Sci., Sofia,
 2017, 11--76, \href{https://arxiv.org/abs/1611.05650}{arXiv:1611.05650}.

\bibitem{Massopust14a}
Massopust P.R., Fractal functions, fractal surfaces, and wavelets, 2nd ed.,
 Elsevier/Academic Press, London, 2016.

\bibitem{Neretin11a}
Neretin Yu.A., Lectures on {G}aussian integral operators and classical groups,
 \textit{EMS Series of Lectures in Mathematics}, \href{https://doi.org/10.4171/080}{European Mathematical Society (EMS)},
 Z\"urich, 2011.

\bibitem{Perelomov86}
Perelomov A., Generalized coherent states and their applications, \textit{Texts and
 Monographs in Physics}, \href{https://doi.org/10.1007/978-3-642-61629-7}{Springer-Verlag}, Berlin, 1986.

\bibitem{Rudin91a}
Rudin W., Functional analysis, 2nd ed., \textit{International Series in Pure and
 Applied Mathematics}, McGraw-Hill, Inc., New York, 1991.

\bibitem{Vasilevski99}
Vasilevski N.L., On the structure of {B}ergman and poly-{B}ergman spaces,
 \href{https://doi.org/10.1007/BF01291838}{\textit{Integral Equations Operator Theory}} \textbf{33} (1999), 471--488.

\bibitem{Weil64a}
Weil A., Sur certains groupes d'op\'erateurs unitaires, \href{https://doi.org/10.1007/BF02391012}{\textit{Acta Math.}}
 \textbf{111} (1964), 143--211.

\bibitem{Zak67a}
Zak J., Finite translations in solid-state physics, \href{https://doi.org/10.1103/PhysRevLett.19.1385}{\textit{Phys. Rev. Lett.}}
 \textbf{19} (1967), 1385--1387.

\end{thebibliography}
\end{document}